\documentclass[sigconf,nonacm]{acmart}

\usepackage{colortbl}
\usepackage[ruled,vlined,linesnumbered]{algorithm2e}
\usepackage{tcolorbox}
\usepackage{wrapfig}
\usepackage{placeins}
\usepackage{stackengine} 
\usepackage{booktabs}
\usepackage{hyperref}
\usepackage{enumitem}
\usepackage{listings}
\usepackage{tikz}
\usepackage{tikzit}

\tikzstyle{basenumber}=[fill=none, draw=none, shape=circle, minimum size=0, inner sep=0]
\tikzstyle{firstbucket}=[fill=none, draw=none, shape=circle, minimum size=0, inner sep=0, text={rgb:red,230;green,159;blue,0}]
\tikzstyle{secondbucket}=[fill=none, draw=none, shape=circle, minimum size=0, inner sep=0, text={rgb:red,86;green,180;blue,233}]
\tikzstyle{thirdbucket}=[fill=none, draw=none, shape=circle, minimum size=0, inner sep=0, text={rgb:red,0;green,158;blue,115}]
\tikzstyle{new style 0}=[fill=white, draw=none, shape=circle]

\tikzstyle{jump}=[->]

\usepackage{multirow}
\usepackage{array}

\usepackage{subcaption}

\usepackage{xcolor}
\usepackage{times}

\newtheorem{observation}{Observation}

\newcommand{\BO}[1]{O\left(#1\right)}
\newcommand{\BT}[1]{\Theta\left(#1\right)}
\renewcommand{\epsilon}{\varepsilon}

\usepackage{xspace}
\newcommand{\algo}[1]{\textsc{#1}}
\newcommand{\ours}{\algo{RD-index}\xspace}

\newcommand{\qrd}{range-duration\xspace}
\newcommand{\qdo}{duration-only\xspace}
\newcommand{\qro}{range-only\xspace}

\newcommand{\Let}[2]{#1 $\leftarrow$ #2}

\newcommand{\domT}{\ensuremath{\Omega^T}}
\newcommand{\args}[2]{#1_1,\dots,#1_{#2}}

\newcommand{\R}{\ensuremath{\mathbf{r}}\xspace} 

\newcommand{\T}{\ensuremath{t}} 
\newcommand{\D}{\ensuremath{d}} 
\newcommand{\TS}{\ensuremath{t_s}}
\newcommand{\TE}{\ensuremath{t_e}}
\newcommand{\rTE}{\ensuremath{r.T_e}}
\newcommand{\rTS}{\ensuremath{r.T_s}}
\newcommand{\ts}{\ensuremath{T_s}}
\newcommand{\te}{\ensuremath{T_e}}
\newcommand{\dmin}{\ensuremath{d_{min}}}
\newcommand{\dmax}{\ensuremath{d_{max}}}
\newcommand{\rT}{\ensuremath{r.T}}

\newcommand{\dataset}[1]{\texttt{#1}}

\newcommand{\argmax}{\operatorname{argmax}}

\newcommand{\Column}{\texttt{column}\xspace}
\newcommand{\Colbounds}{\texttt{col\_minstart}\xspace}
\newcommand{\Colboundsi}{\texttt{col\_minstart[i]}\xspace}
\newcommand{\Cell}{\texttt{cell}\xspace}
\newcommand{\Cellbounds}{\texttt{cell\_mindur}\xspace}
\newcommand{\Cellboundsi}{\texttt{cell\_mindur[i]}\xspace}
\newcommand{\Cellboundsij}{\texttt{cell\_mindur[i][j]}\xspace}
\newcommand{\Grid}{\texttt{grid}\xspace}
\newcommand{\Gridij}{\texttt{grid[i][j]}\xspace}
\newcommand{\LatestEnds}{\texttt{col\_maxend}\xspace}
\newcommand{\LatestEndsi}{\texttt{col\_maxend[i]}\xspace}
\newcommand{\LatestEndsPrev}{\texttt{col\_maxend[i-1]}\xspace}
\newcommand{\MaxDurations}{\texttt{cell\_maxdur}\xspace}
\newcommand{\MaxDurationsi}{\texttt{cell\_maxdur[i]}\xspace}
\newcommand{\MaxDurationsij}{\texttt{cell\_maxdur[i][j]}\xspace}
\newcommand{\nextsubseq}{\textsc{NextSubseq}\xspace}
\newcommand{\buildindex}{\textsc{BuildIndex}\xspace}
\newcommand{\query}{\textsc{Query}\xspace}

\newcommand{\btree}{\algo{B-Tree}\xspace}
\newcommand{\rtree}{\algo{R*-Tree}\xspace}
\newcommand{\itree}{\algo{Interval-Tree}\xspace}
\newcommand{\gfile}{\algo{Grid-File}\xspace}
\newcommand{\pindex}{\algo{Period-Index$\star$}\xspace}
\newcommand{\rdtd}{\algo{RD-index-td}\xspace}
\newcommand{\rddt}{\algo{RD-index-dt}\xspace}
\newcommand{\rd}{\algo{RD-index}\xspace}

\definecolor{firstcolor}{HTML}{FDE725}
\definecolor{secondcolor}{HTML}{7AD151}
\definecolor{thirdcolor}{HTML}{22A884}
\definecolor{fourthcolor}{HTML}{2A788E}
\definecolor{fifthcolor}{HTML}{414487}
\definecolor{sixthcolor}{HTML}{440154}

\begin{document}

\title{Indexing Temporal Relations for Range-Duration Queries}

\author{Matteo Ceccarello}
\affiliation{
 Faculty of Computer Science\\ Free University of Bozen-Bolzano, Italy
}
\email{mceccarello@unibz.it}

\author{Anton Dign\"os}
\affiliation{
 Faculty of Computer Science\\ Free University of Bozen-Bolzano, Italy
}
\email{dignoes@inf.unibz.it}

\author{Johann Gamper}
\affiliation{
 Faculty of Computer Science\\ Free University of Bozen-Bolzano, Italy
}
\email{gamper@inf.unibz.it}

\author{Christina Khnaisser}
\affiliation{
 Facult\'e de m\'edecine et des sciences de la sant\'e\\ Universit\'e de Sherbrooke, Canada
}
\email{christina.khnaisser@usherbrooke.ca}

\begin{abstract}
  Temporal information plays a crucial role in many database
  applications, however support for queries on such data is limited.
  We present an index structure, termed \rd, to support
  \emph{range-duration queries} over interval timestamped relations,
  which constrain both the \emph{range} of the tuples' positions on
  the timeline and their \emph{duration}.  \rd is a grid structure in
  the two-dimensional space, representing the position on the timeline
  and the duration of timestamps, respectively.  Instead of using a
  regular grid, we consider the data distribution for the construction
  of the grid in order to ensure that each grid cell contains
  approximately the same number of intervals.  \rd features provable
  bounds on the running time of all the operations, allow for a simple
  implementation, and supports very predictable query performance.
  We benchmark our solution on a variety of datasets and query
  workloads, investigating both the query rate and the behavior of the
  individual queries.  The results show that \rd performs better than
  the baselines on range-duration queries, for which it is explicitly
  designed.  Furthermore, it outperforms specialized indexes also on
  workloads containing queries constraining either only the duration
  or the range.
\end{abstract}

\maketitle

\section{Introduction}

Temporal information plays a crucial role in many database
applications: in fact, many database management systems and the SQL
standard~\cite{DBLP:journals/sigmod/KulkarniM12} provide automated version control of the data and time travel
facilities, allowing to efficiently access past history.  Past
research mainly concentrated on efficient solutions for important
temporal operators, such as temporal
aggregation~\cite{DBLP:conf/icde/KlineS95,DBLP:journals/tkde/MoonLI03,DBLP:conf/edbt/BohlenGJ06,DBLP:conf/ssd/PiatovH17},
temporal joins~\cite{DBLP:conf/icde/PiatovHD16,DBLP:journals/vldb/BourosMTT21,DBLP:journals/vldb/DignosBGJM21}, and time
travel~\cite{DBLP:conf/sigmod/KaufmannMVFKFM13} queries.
All these approaches consider only the position of intervals along the
timeline, ignoring another important aspect, namely the
\emph{duration} of intervals.
As a result, index structures to support more general selection
queries that constrain both the duration and the position in time of
intervals have been
missing~\cite{DBLP:journals/sigmod/KulkarniM12,DBLP:journals/tods/DignosBGJ16,DBLP:conf/ebiss/BohlenDGJ17}
until recently~\cite{BehrendDGSVRK_SSTD19_period-index}.
In many application domains, however, both aspects of temporal
information are useful to formulate queries.

\begin{example}\label{ex:motivation}
  As a concrete use case, consider the use of antibiotics in
  healthcare. Antibiotic resistance is a world challenge, and the
  emergence of new resistance factors is very difficult to monitor and
  to predict due to the diversity of antibiotic usage and events
  (e.g., environment, species evolution, medical practices,
  etc.)~\cite{Hayashi2011,Larsson2021}.  Selecting the most
  appropriate antibiotic and the appropriate treatment duration is an
  essential step to reduce antibiotic
  resistance~\cite{Hayashi2011}. Thus, defining guidelines for the
  duration of antibiotic treatments, measuring the adherence to these
  guidelines, and developing stewardship tools regarding antibiotics
  usage can help clinicians in choosing the optimal treatment
  considering the patient's medical
  history~\cite{Shapiro2021,Larsson2021}.  Such measures should be
  implemented at a national level in order to monitor and audit
  antibiotic resistance on a larger scale. In this context, the
  following types of temporal queries are frequent:
  \begin{itemize}
  \item[$Q1$:] \emph{``Find all antibiotics prescriptions from October
      1, 2016 to March 31, 2017.''}
  \item[$Q2$:] \emph{``Find all antibiotics prescriptions with a
      treatment duration between 5 and 8 days.''}
  \item[$Q3$:] \emph{``Find all antibiotics prescriptions from October
      1, 2016 to March 31, 2017, with a treatment duration between 1
      and 2 weeks.''}
  \end{itemize}
  %
  %
  The first query $Q1$ retrieves tuples based on the position of the
  events on the timeline; we call it \emph{range query}.  In contrast,
  query $Q2$ imposes constraints on the duration of matching events,
  and we call it \emph{duration query}.  Finally, query $Q3$
  constrains both types of information; we call it
  \emph{range-duration query}.
  This type of queries can be found and have been reported as a
  primitive operation in other application scenarios that deal with
  interval data, e.g., in air traffic
  analysis~\cite{DBLP:conf/gis/SchullerBM10,DBLP:conf/adbis/BehrendMSW09,DBLP:conf/adbis/SchullerSB12},
  event detection for video
  surveillance~\cite{DBLP:conf/cikm/PersiaBH17}, or the analysis of
  clinical data~\cite{DBLP:conf/adbis/BehrendSXFGLCG14}.
\end{example}

Existing index structures typically support only one of the two
aspects, either the position of the interval on the timeline or the
duration of the interval.  For instance, the well-known relational
interval tree~\cite{DBLP:conf/vldb/KriegelPS00} is optimized for
efficiently determining temporal relationships between intervals but
not interval lengths. In the worst case, the entire index tree must be
traversed if a query solely contains restrictions on the interval
length.  Alternatively, the duration of the intervals can be indexed
straightforwardly using a classic data structure such as a B-tree.  In
this case, however, queries constraining only the range of the
intervals will need to traverse the entire tree.
Combining the two indexes is typically inefficient since a query can
have different selectivities in the two dimensions.
Therefore, to efficiently support workloads involving a mix of all
three types of queries mentioned above we seek a new index structure
that supports both dimensions at the same time.

In this paper, we introduce \rd, a novel two-dimensional data
structure that indexes time intervals both on their position on the
timeline and their duration.  Our index structure partitions the
intervals in a grid according to their start times and
durations. Rather than constructing a regular grid, the boundaries
between the grid cells are determined by taking into account the data
distribution.  Such a strategy ensures that each cell contains
approximately the same number of intervals, with the exception of some
edge cases if the distribution of the intervals is extremely skewed.
The uniform distribution of the data over all grid cells allows to
obtain very predictable query times, which are proportional to the
selectivity of the query.  We prove that the time for answering a
range-duration query with \rd is
$O(\frac{n}{s^2} \log \frac{n}{s} + \frac{n}{s} + s^2 + k)$, where $n$
is the size of the input relation, $k$ is the number of intervals
matching the query predicate, and $s$ is the page size.  The index can
be constructed in $O(n \log n)$ time.  The page size $s$ is the only
parameter of our index structure, and it is independent of the data
distribution. The index structure also lends itself to a rather simple
implementation.  While being explicitly designed to address
range-duration queries, \rd also supports \qro and \qdo queries
efficiently.

We present the results of a detailed experimental evaluation. The
results show that the overhead introduced by the data structure is
indeed negligible and that the running time in practice is largely
proportional to the selectivity of the query.  This is in contrast to
the competitors we compare to.  On \qrd queries we find that \rd
clearly outperforms the competitors. On mixed workloads comprising all
three types of queries, we find that \rd outperforms the competitors
in the vast majority of the workloads, even for cases for which
specialized solutions exist. 

Our contributions can be summarized as follows:
\begin{itemize}
\item We describe \rd, a novel index structure that supports temporal
  queries involving both the duration and the range of time intervals.
\item We prove bounds on the performance of $\rd$, which can be tuned
  with a single page size parameter~$s$.
\item We provide an extensible open source implementation, which we
  benchmark against state-of-the-art competitors, showing
  significantly better performance across several workloads.
  In particular, we show how \rd can efficiently handle mixed workloads where all three types of queries coexist.
\end{itemize}

The rest of the paper is organized as follows.  In
Section~\ref{sec:related-work} we review the state of the art, before
laying out the fundamental concepts underlying our approach in
Section~\ref{sec:preliminaries}.  Our data structure is introduced in
Section~\ref{sec:rd-index-structure}, and the complexity of all
operations is analyzed in Section~\ref{sec:analysis}.  Experimental
results are presented in Section~\ref{sec:experiments}, before drawing
our conclusions in Section~\ref{sec:conclusions}.

\section{Related work}
\label{sec:related-work}

The type of selection queries we are studying in this paper are
queries with a conjunctive predicate, where one predicate in the
conjunction restricts the position of intervals on the time line and
the second the duration of intervals. In this section, we review
indexing structures that are (partially) suitable for such selection
queries, and also review structures similar to our approach that are
used for interval joins.

There are several approaches devoted to indexing interval timestamped data.
Edelsbrunner's Interval Tree~\cite{edelsbrunner80} is one of the most
popular indexing structure for intervals. It is asymptotically optimal
for selection queries involving the overlap predicate, and there
exists an implementation using standard relational database technology
based on B-tree indexes~\cite{DBLP:conf/vldb/KriegelPS00}.  A
shortcoming of the interval tree is that, in contrast to our indexing
structure, it does not provide a mechanism to restrict the duration of
intervals, and thus can only solve one part of a range-duration query.
A similar indexing structure is the segment tree~\cite{deBerg00}.  It
builds disjoint segments over intervals at the leaf level using all
start and end points in a relation, and recursively merges segments in
intermediate nodes of the tree. This data structure was originally
designed for point queries over intervals (also known as time travel queries), i.e., for
retrieving all intervals overlapping a given time point. The segment
tree also supports selection queries with the overlap predicate given
a query interval, albeit in this case a duplicate elimination step for
intervals retrieved multiple times is required. Another index structure that support time travel queries is the timeline index~\cite{DBLP:conf/sigmod/KaufmannMVFKFM13,DBLP:journals/pvldb/Kaufmann13}. The timeline index stores the start and end points of intervals in an event list in sorted order and allows to retrieve all tuples that overlap a given time point by scanning through the event list and discarding tuples that ended before the given time point. To avoid scanning through the entire event list, the index maintains regular checkpoints that store all tuples that overlap the time point of the checkpoint. 
Similarly to the
interval tree, the segment tree and timeline index do not support to restrict the duration of intervals. 

In a two dimensional space, intervals can be represented as 2D points,
where one dimension is the start point and the second dimension is
either the end point or the duration of an interval. In such a space,
a selection query with the overlap predicate corresponds to a
selection query over a 2D area. For this, multidimensional indices can
be used.
R-trees~\cite{DBLP:conf/sigmod/BeckmannKSS90,DBLP:journals/talg/ArgeBHY08}
are multidimensional indices that group objects in a multidimensional
space using minimum bounding (hyper) rectangles. Quadtrees and
octrees~\cite{DBLP:journals/acta/FinkelB74,Ulrich00,Samet05}
recursively divide the space into partitions and place objects into
the best fitting partition according to some criteria.  All of the
aforementioned indexes are linked data structures, which suffer from
poor locality, both when implemented in-memory and on disk.  In
contrast, our \ours can be implemented by means of simple arrays, and
thus enjoys high cache locality.

Another multidimensional index structure is the grid
file~\cite{DBLP:journals/tods/NievergeltHS84}. In the context of time
intervals, the idea would be to partition the span of durations and of
starting times into cells of equal width, thus allowing efficient
access to both dimensions. The main drawback of this data structure is
that in case of skewed data distributions the load of the cells is
unequal, which might significantly harm the performance.

A recent approach to multidimensional indexing is that of
\emph{learned indexes}~\cite{DBLP:conf/sigmod/KraskaBCDP18}: the
proposition is that index structures are \emph{models} mapping keys to
records. Therefore machine learning models can be used to provide this
mapping, in lieu of the classic data structures. In particular, Flood
and
Tsunami~\cite{DBLP:conf/sigmod/NathanDAK20,DBLP:journals/corr/abs-2006-13282}
use Recursive Model Index~\cite{DBLP:conf/sigmod/KraskaBCDP18} and a
variant of decision trees to model the position of records in the
database, adapting to the distribution of the data and of the query
workload.
Our approach shares some ideas with this line of work, namely adapting
to the data distribution by means of the conditional cumulative
distribution function.  However, while our index supports both
insertions and deletions, both Flood and Tsunami are tailored at
read-only workloads.  Furthermore, we prove bounds on the worst case
running times for all the operations,
while~\cite{DBLP:conf/sigmod/NathanDAK20,DBLP:journals/corr/abs-2006-13282}
provide an empirical evaluation.  Finally, our approach is arguably
simpler, in that it is based just on sorting and iterating through
records.

Very recently,
Behrend et al.~\cite{BehrendDGSVRK_SSTD19_period-index} proposed an
index, named \pindex, that explicitly supports range-duration
queries. The index partitions the time domain in \emph{buckets}.  An
interval is assigned to all buckets it intersects with.  Within each
bucket, intervals are further partitioned in \emph{levels} according
to their duration, with the minimum duration indexed within each level
decreasing geometrically. To support efficient indexing along the
start time dimension, each level is further partitioned in the time
domain. This data structure is adaptive to the distribution of start
times, while it assumes a Zipf-like distribution for the duration of
the intervals.  Our index structure removes this assumption, thus
supporting datasets with arbitrary distributions of the tuples'
duration.  Furthermore, our index features only one data-independent
parameter, instead of the two data-dependent parameters of \pindex,
and it allows to control whether to index first by duration or time.
Moreover, we do not replicate intervals in the index, yielding a
significantly smaller structure, thereby avoiding the consequent
possible performance degradation.
Finally, \pindex does not support updates of the index.

In recent years, algorithms for interval joins, which can be seen as a sequence of range queries, have been actively
studied. Approaches based on the timeline
index~\cite{DBLP:conf/sigmod/KaufmannMVFKFM13} process sets of
intervals as sorted event lists of their start and end points.  The
interval join is computed by scanning these event lists in an
interleaved fashion, thereby keeping and joining sets of active
intervals, i.e., intervals whose start has been encountered but not
the end point.  To improve the performance of the original linked list
data structure for storing active intervals, a gapless hash map has
been proposed in~\cite{DBLP:conf/icde/PiatovHD16} that provides a
higher performance for scanning active intervals. The same idea has
also been extended for joins using general Allen's
predicates~\cite{DBLP:journals/vldb/PiaDHP21} rather than only overlap
predicates.  The works
in~\cite{DBLP:journals/pvldb/BourosM17,DBLP:journals/vldb/BourosMTT21,DBLP:journals/vldb/DignosBGJM21}
compute an interval join using sorting and backtracking. First, the
input relations are sorted by start time and then an interleaving
merge-join is performed to compute the temporal join between the two
relations. Other
approaches~\cite{DBLP:conf/sigmod/DignosBG14,DBLP:journals/vldb/CafagnaB17}
for the interval join are based on partitioning intervals according to
their position and then produce the join result by joining relevant
partitions.  While all these approaches for the interval join provide
mechanisms to join overlapping intervals, in contrast to our work they
are not applicable for general selection queries as they always
require to read the entire relations.  Moreover, the duration of
intervals is not considered at all in these works.

\section{Preliminaries}
\label{sec:preliminaries}

We assume a linearly ordered, discrete time domain, $\domT$.  A time
interval is a set of contiguous time points, and $\T = [\TS,\TE)$
denotes the closed-open interval of points from $\TS$ to $\TE$.  We
use $|t| = \TE - \TS$ to denote the duration of time interval $\T$ and
$\T \cap \T'$ to denote the set of time points shared by two intervals
$\T$ and $\T'$, which, if not empty, is itself an interval.
The schema of a temporal relation is given by $R = (\args{A}{m},T)$,
where $\args{A}{m}$ are the non-temporal attributes with domains
$\Omega_i$ and $T$ is the time interval attribute with domain
$\domT \times \domT$ representing, for instance, the tuple's valid
time.  A temporal relation $\R$ with schema $R$ is a finite set of
tuples, where each tuple has a value in the appropriate domain for
each attribute in the schema. We use $r.A_i$ to denote the value of
attribute $A_i$ in tuple $r$, and $\rT = [\rTS, \rTE)$ to refer to its
time interval. 

The index we propose efficiently supports the three following types of
temporal queries (defined as
in~\cite{BehrendDGSVRK_SSTD19_period-index}).

\begin{definition}[Range query]
  Given a temporal interval $\T = [\TS, \TE)$ and a temporal relation
  $\R$, a \emph{range query} is defined as
  \[
    Q(\R, \T) = \left\{r \in \R : \rT \cap \T \ne \emptyset
    \right\}
  \]
\end{definition}

\begin{definition}[Duration query]
  Given a duration interval $\D = [\dmin, \dmax]$ and a temporal
  relation $\R$, a \emph{duration query} is defined as
  \[
    Q(\R, \D) = \left\{r \in \R : |\rT| \in [\dmin, \dmax] \right\}
  \]
\end{definition}

\begin{definition}[Range-duration query]
  Given a temporal interval $\T = [\TS, \TE)$, a duration interval
  $\D = [\dmin, \dmax]$, and a temporal relation $\R$, a
  \emph{range-duration query} is defined as
  \[
    Q(\R, \T, \D) = \left\{r \in \R : \rT \cap \T \ne \emptyset \wedge
      |\rT| \in [\dmin, \dmax] \right\}
  \]
\end{definition}

A range query retrieves all tuples whose time interval intersects with the
query range $\T$.  A duration query retrieves all tuples whose
time interval has a duration that is between $d_{min}$ and $d_{max}$. A
range-duration query is a combination of the former two.

\begin{example}\label{example:queries}
  As a running example, we consider real-world drug prescriptions from
  the MIMICIII open source database~\cite{mimiciii} (cf.\ use case in
  Example~\ref{ex:motivation}).  It stores antibiotic prescriptions,
  characterized by a start date and an end date of the prescription
  and the duration of the treatment.  An excerpt of four tuples is
  shown in Figure~\ref{tab:example}.
  Consider the following range-duration query: Retrieve all
  prescriptions in the period from June 15 to July 15 with a treatment
  duration between 5 and 15 days.  Figure~\ref{fig:example-intervals}
  shows a graphical representation of the relation and the query,
  where the time intervals are drawn by thick solid horizontal lines.
  The red area indicates the constraint on the position.  The red
  segment below of each timestamp interval denotes the duration
  constraint, where the dotted red line indicates the range of minimum
  and maximum duration. Hence, an interval satisfies the range
  constraint if it intersects with the red area, and it satisfies the
  duration constraint if its end point is within the dotted red line.
  Tuple $r_1$ satisfies only the duration constraint, tuple $r_2$
  satisfies neither constraint, tuple $r_3$ satisfies only the range
  constraint, and tuple $r_4$ satisfies both.
\end{example}

\begin{figure}
  \begin{subfigure}{\columnwidth}
    \centering
    \begin{tabular}[t]{r|lrr|r}
      \cline{2-4}
     & drug &  $\ts$ &  $\te$ & {\color{gray} duration}\\
      \cline{2-4}
    $r_1$ & Amoxicillin & June 08 & June 14 & {\color{gray} (6 days)}\\
    $r_2$ & Amoxicillin & June 10 & June 12 & {\color{gray} (2 days)}\\
    $r_3$ & Ceftriaxone & June 20 & July 05 & {\color{gray} (15 days)}\\
    $r_4$ & Levofloxacin & June 24 & July 04 & {\color{gray} (10 days)}\\
      \cline{2-4}
    \end{tabular}

    \caption{Sample of relation with antibiotic prescriptions\label{tab:example}}
  \end{subfigure}

  \medskip
  \begin{subfigure}{\columnwidth}
    \centering
    \begin{tikzpicture}[x=1pt,y=1pt]
    \definecolor{fillColor}{RGB}{255,255,255}
    \path[use as bounding box,fill=fillColor,fill opacity=0.00] (0,0) rectangle (238.49, 72.27);
    \begin{scope}
    \path[clip] (  3.00, 13.88) rectangle (238.49, 72.27);
    \definecolor{drawColor}{gray}{0.92}

    \path[draw=drawColor,line width= 0.3pt,line join=round] ( 23.22, 13.88) --
      ( 23.22, 72.27);

    \path[draw=drawColor,line width= 0.3pt,line join=round] ( 99.34, 13.88) --
      ( 99.34, 72.27);

    \path[draw=drawColor,line width= 0.3pt,line join=round] (170.70, 13.88) --
      (170.70, 72.27);

    \path[draw=drawColor,line width= 0.6pt,line join=round] ( 61.28, 13.88) --
      ( 61.28, 72.27);

    \path[draw=drawColor,line width= 0.6pt,line join=round] (137.40, 13.88) --
      (137.40, 72.27);

    \path[draw=drawColor,line width= 0.6pt,line join=round] (204.00, 13.88) --
      (204.00, 72.27);
    \definecolor{fillColor}{RGB}{255,0,0}

    \path[fill=fillColor,fill opacity=0.02] ( 61.28, 16.53) rectangle (204.00, 69.62);

    \path[fill=fillColor,fill opacity=0.02] ( 61.28, 16.53) rectangle (204.00, 69.62);

    \path[fill=fillColor,fill opacity=0.02] ( 61.28, 16.53) rectangle (204.00, 69.62);

    \path[fill=fillColor,fill opacity=0.02] ( 61.28, 16.53) rectangle (204.00, 69.62);
    \definecolor{drawColor}{RGB}{0,0,0}

    \path[draw=drawColor,line width= 1.7pt,line join=round] ( 37.49, 35.95) -- ( 47.01, 35.95);

    \path[draw=drawColor,line width= 1.7pt,line join=round] ( 85.07, 48.90) -- (156.43, 48.90);
    \definecolor{drawColor}{RGB}{0,0,255}

    \path[draw=drawColor,line width= 1.7pt,line join=round] (104.09, 61.85) -- (151.67, 61.85);
    \definecolor{drawColor}{RGB}{0,0,0}

    \path[draw=drawColor,line width= 1.7pt,line join=round] ( 27.98, 23.01) -- ( 56.52, 23.01);

    \node[text=drawColor,anchor=base,inner sep=0pt, outer sep=0pt, scale=  0.85] at ( 39.03, 39.84) {$r_2$};

    \node[text=drawColor,anchor=base,inner sep=0pt, outer sep=0pt, scale=  0.85] at ( 86.60, 52.78) {$r_3$};

    \node[text=drawColor,anchor=base,inner sep=0pt, outer sep=0pt, scale=  0.85] at (105.63, 65.73) {$r_4$};

    \node[text=drawColor,anchor=base,inner sep=0pt, outer sep=0pt, scale=  0.85] at ( 29.51, 26.89) {$r_1$};
    \definecolor{drawColor}{RGB}{255,0,0}
    \definecolor{fillColor}{RGB}{255,0,0}

    \path[draw=drawColor,draw opacity=0.70,line width= 0.4pt,line join=round,line cap=round,fill=fillColor,fill opacity=0.70] ( 37.49, 33.36) circle (  1.43);

    \path[draw=drawColor,draw opacity=0.70,line width= 0.4pt,line join=round,line cap=round,fill=fillColor,fill opacity=0.70] ( 85.07, 46.31) circle (  1.43);

    \path[draw=drawColor,draw opacity=0.70,line width= 0.4pt,line join=round,line cap=round,fill=fillColor,fill opacity=0.70] (104.09, 59.26) circle (  1.43);

    \path[draw=drawColor,draw opacity=0.70,line width= 0.4pt,line join=round,line cap=round,fill=fillColor,fill opacity=0.70] ( 27.98, 20.42) circle (  1.43);

    \path[draw=drawColor,draw opacity=0.70,line width= 0.7pt,line join=round] ( 37.49, 33.36) -- ( 61.28, 33.36);

    \path[draw=drawColor,draw opacity=0.70,line width= 0.7pt,line join=round] ( 85.07, 46.31) -- (108.85, 46.31);

    \path[draw=drawColor,draw opacity=0.70,line width= 0.7pt,line join=round] (104.09, 59.26) -- (127.88, 59.26);

    \path[draw=drawColor,draw opacity=0.70,line width= 0.7pt,line join=round] ( 27.98, 20.42) -- ( 51.76, 20.42);

    \path[draw=drawColor,draw opacity=0.70,line width= 0.9pt,dash pattern=on 1pt off 3pt ,line join=round] ( 61.28, 33.36) -- ( 99.34, 33.36);

    \path[draw=drawColor,draw opacity=0.70,line width= 0.9pt,dash pattern=on 1pt off 3pt ,line join=round] (108.85, 46.31) -- (146.91, 46.31);

    \path[draw=drawColor,draw opacity=0.70,line width= 0.9pt,dash pattern=on 1pt off 3pt ,line join=round] (127.88, 59.26) -- (165.94, 59.26);

    \path[draw=drawColor,draw opacity=0.70,line width= 0.9pt,dash pattern=on 1pt off 3pt ,line join=round] ( 51.76, 20.42) -- ( 89.82, 20.42);
    \end{scope}
    \begin{scope}
    \path[clip] (  0.00,  0.00) rectangle (238.49, 72.27);
    \definecolor{drawColor}{RGB}{169,169,169}

    \path[draw=drawColor,line width= 0.6pt,line join=round] (  3.00, 13.88) --
      (238.49, 13.88);
    \end{scope}
    \begin{scope}
    \path[clip] (  0.00,  0.00) rectangle (238.49, 72.27);
    \definecolor{drawColor}{gray}{0.30}

    \node[text=drawColor,anchor=base,inner sep=0pt, outer sep=0pt, scale=  0.96] at ( 61.28,  1.87) {Jun 15};

    \node[text=drawColor,anchor=base,inner sep=0pt, outer sep=0pt, scale=  0.96] at (137.40,  1.87) {Jul 01};

    \node[text=drawColor,anchor=base,inner sep=0pt, outer sep=0pt, scale=  0.96] at (204.00,  1.87) {Jul 15};
    \end{scope}
    \end{tikzpicture}

    \caption{Range-duration query\label{fig:example-intervals}}
  \end{subfigure}
  \caption{Running example.}
\end{figure}
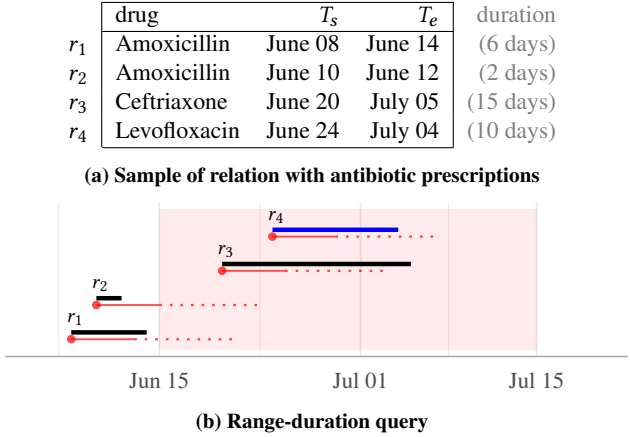

\section{The \ours Structure}
\label{sec:rd-index-structure}

\subsection{Overview}
\label{sec:overview}

\begin{figure*}[t]
  \includegraphics[width=\textwidth]{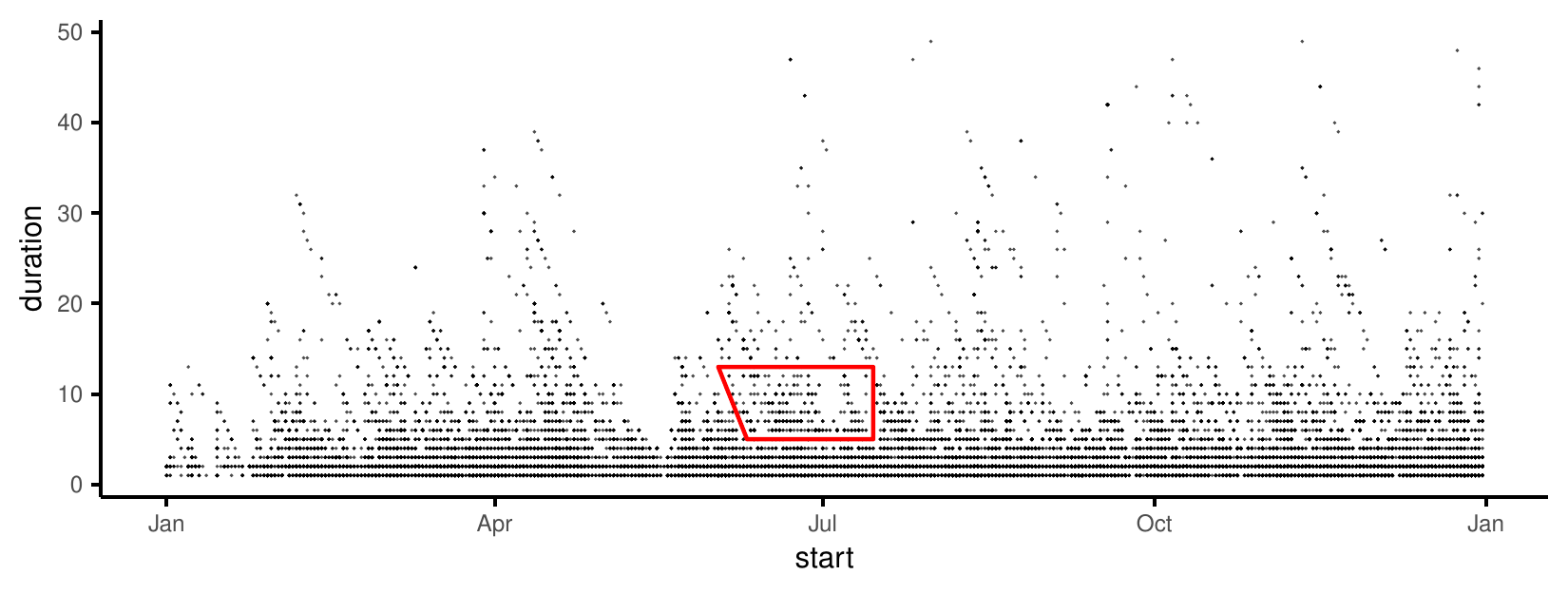}
  \caption{Drug prescription dataset in the start time/duration plane. 
	The gray dots are tuples.
    The red polygon is the range-duration query of
    Example~\ref{example:queries}.
    \label{fig:example-plane}
  }
\end{figure*}

An interval can be completely described by its starting
time and duration, alternatively, to its starting and ending time.
Therefore, a set of temporal intervals can be represented as a
set of points in a two dimensional space, with one coordinate being
the start time and the other the duration.
A range-duration query $Q(\R, \T, \D)$ with $\T = [\TS, \TE)$ and
$\D = [\dmin, \dmax]$ in such a space is represented by a polygon
containing all tuples $r \in \R$ such that

\begin{align}
  & |\rT| \in [\dmin, \dmax] \label{eq:duration-condition} \\
  & \rTE > \TS                  \label{eq:start-condition}     \\
  & \rTS < \TE                  \label{eq:end-condition}
\end{align}

\begin{example}
  Figure~\ref{fig:example-plane} shows a sample of 31\,192 drug
  prescriptions from the MIMICIII dataset in the two dimensional
  space, where each point represents a timestamp interval of a
  tuple. The red polygon indicates the query of
  Example~\ref{example:queries}, i.e., the prescriptions from June 15
  to July 15 with a duration between 5 and 15 days.  Note that the
  points are not uniformly distributed, hence efficiently indexing
  this dataset requires to be adaptive in both the start time and
  duration dimension.
\end{example}

The \ours we are presenting is a two-dimensional grid, partitioning
the tuples into disjoint buckets according to the start time and the
duration of the intervals.  The boundaries between cells are defined
by using the empirical cumulative distribution function of the tuples'
duration and the starting times, so that each cell contains
approximately the same number of intervals, which corresponds to the
\emph{page size} $s$ of our index.  This allows the index to adapt to
the distribution of the input and to different scenarios.  In a main
memory scenario, the parameter $s$ could be set such that a cell fits
in a cache line.  In an external memory setting, it might be set to
the disk block size.

To simplify the presentation, in the following we will focus on the
timestamp interval attribute $T$ of relation $\R$: in the discussion
we assume that each interval being inserted in the index is associated
with a reference to the tuple in $\R$ it belongs to.

\subsection{Index Construction}
\label{sec:index-construction}

The grid structure of the \ours partitions an array of tuples first
along either the start time or duration dimension of the intervals,
and then along the other.  The choice of which dimension to index
first may impact the performance of the index, depending on the query
workload and the data distribution (cf.\
Section~\ref{sec:experiments}.  In the following we assume that the
start time dimension is partitioned first, followed by the duration
dimension.  All the descriptions, considerations, and proofs also hold
with the dimensions swapped.

Before describing the algorithm to build the index, we present the
subroutine \nextsubseq, which partitions an array of tuples that is
sorted according to a given key function.  We will use this subroutine
to determine columns and cells of the grid structure, using first the
start time and then duration as keys. Given an index $h$ and a size
parameter $b$, \nextsubseq returns a subsequence starting at $h$ that
either contains at most $b$ tuples, or contains tuples that all share
the same key. Additionally, all the tuples with the same key are part
of the same subsequence. The pseudocode is reported in
Algorithm~\ref{alg:next-breakpoint} and works as follows. Starting
from index position $h$, if there are fewer than $b$ elements after
$h$ then we return all the tuples from $h$ onwards.  Otherwise, we
look at the tuple at position $h' = h + b$ and consider two
cases\footnote{Therefore $h'$ is the end
  index of the subsequence, \emph{non-inclusive}.}. If the tuples at position $h$ and $h'$
have the same key, then we scan forward until the first tuple with a
different key occurs
(lines~\ref{line:heavy-bucket}--\ref{line:heavy-bucket-end}).
Otherwise, we scan backward until two consecutive tuples have
different keys
(lines~\ref{line:split-bucket}--\ref{line:split-bucket-end}).  In both
cases, the rationale is to avoid splitting runs of same-key tuples
between different subsequences.

Note that Algorithm~\ref{alg:next-breakpoint} might return a
subsequence with more than $b$ elements if and only if all share the
same key.  In such case we deem the returned subsequence \emph{heavy}, otherwise we deem it \emph{light}.

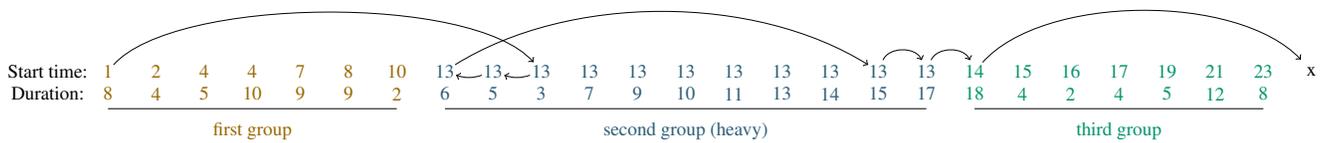
\begin{figure*}[t]
  \scalebox{.8}{
    \tikzstyle{every picture}=[baseline=-0.25em,scale=0.5,xscale=2]
    \begin{tikzpicture}[xscale=.8]
      \begin{pgfonlayer}{nodelayer}
        \node [style=firstbucket] (2) at (0, 0) {1};
        \node [style=firstbucket] (4) at (1, 0) {2};
        \node [style=firstbucket] (5) at (2, 0) {4};
        \node [style=firstbucket] (6) at (3, 0) {4};
        \node [style=firstbucket] (7) at (4, 0) {7};
        \node [style=firstbucket] (8) at (5, 0) {8};
        \node [style=firstbucket] (9) at (6, 0) {10};
        \node [style=secondbucket] (10) at (9, 0) {13};
        \node [style=secondbucket] (11) at (8, 0) {13};
        \node [style=secondbucket] (12) at (7, 0) {13};
        \node [style=secondbucket] (13) at (10, 0) {13};
        \node [style=secondbucket] (14) at (11, 0) {13};
        \node [style=secondbucket] (15) at (12, 0) {13};
        \node [style=secondbucket] (16) at (13, 0) {13};
        \node [style=secondbucket] (17) at (14, 0) {13};
        \node [style=secondbucket] (18) at (15, 0) {13};
        \node [style=secondbucket] (19) at (17, 0) {13};
        \node [style=secondbucket] (20) at (16, 0) {13};
        \node [style=thirdbucket] (21) at (18, 0) {14};
        \node [style=thirdbucket] (22) at (19, 0) {15};
        \node [style=thirdbucket] (23) at (20, 0) {16};
        \node [style=thirdbucket] (24) at (21, 0) {17};
        \node [style=thirdbucket] (25) at (22, 0) {19};
        \node [style=thirdbucket] (26) at (23, 0) {21};
        \node [style=thirdbucket] (27) at (24, 0) {23};
        \node [style=none] (28) at (0, -1.25) {};
        \node [style=none] (29) at (6, -1.25) {};
        \node [style=none] (30) at (7, -1.25) {};
        \node [style=none] (31) at (17, -1.25) {};
        \node [style=none] (32) at (18, -1.25) {};
        \node [style=none] (33) at (24, -1.25) {};
        \node [style=none] (34) at (3, -2) {};
        \node [style=firstbucket] (35) at (3, -2) {first group};
        \node [style=secondbucket] (36) at (12, -2) {second group (heavy)};
        \node [style=thirdbucket] (39) at (21, -2) {third group};
        \node [style=new style 0] (40) at (25, 0) {x};
        \node [style=firstbucket] (41) at (0, -0.75) {8};
        \node [style=firstbucket] (42) at (1, -0.75) {4};
        \node [style=firstbucket] (43) at (2, -0.75) {5};
        \node [style=firstbucket] (44) at (3, -0.75) {10};
        \node [style=firstbucket] (45) at (4, -0.75) {9};
        \node [style=firstbucket] (46) at (5, -0.75) {9};
        \node [style=firstbucket] (47) at (6, -0.75) {2};
        \node [style=secondbucket] (48) at (9, -0.75) {3};
        \node [style=secondbucket] (49) at (8, -0.75) {5};
        \node [style=secondbucket] (50) at (7, -0.75) {6};
        \node [style=secondbucket] (51) at (10, -0.75) {7};
            \node [style=secondbucket] (52) at (11, -0.75) {9};
        \node [style=secondbucket] (53) at (12, -0.75) {10};
        \node [style=secondbucket] (54) at (13, -0.75) {11};
        \node [style=secondbucket] (55) at (14, -0.75) {13};
        \node [style=secondbucket] (56) at (15, -0.75) {14};
        \node [style=secondbucket] (58) at (16, -0.75) {15};
        \node [style=secondbucket] (57) at (17, -0.75) {17};
        \node [style=thirdbucket] (59) at (18, -0.75) {18};
        \node [style=thirdbucket] (60) at (19, -0.75) {4};
        \node [style=thirdbucket] (61) at (20, -0.75) {2};
        \node [style=thirdbucket] (62) at (21, -0.75) {4};
        \node [style=thirdbucket] (63) at (22, -0.75) {5};
        \node [style=thirdbucket] (64) at (23, -0.75) {12};
        \node [style=thirdbucket] (65) at (24, -0.75) {8};
        \node [style=none] (66) at (-1.25, 0) {};
        \node [style=none] (67) at (-1.25, 0) {Start time:};
        \node [style=none] (68) at (-1.25, -0.75) {};
        \node [style=none] (69) at (-1.25, -0.75) {Duration:};
      \end{pgfonlayer}
      \begin{pgfonlayer}{edgelayer}
        \draw [style=jump, in=120, out=60, looseness=0.75] (2) to (10);
        \draw [style=jump, bend left=45] (12) to (20);
        \draw [style=jump, bend left=75, looseness=1.50] (20) to (19);
        \draw [style=jump, bend left=75, looseness=1.50] (19) to (21);
        \draw [style=jump, bend left] (10) to (11);
        \draw [style=jump, bend left] (11) to (12);
        \draw (32.center) to (33.center);
        \draw (30.center) to (31.center);
        \draw (28.center) to (29.center);
        \draw [style=jump, bend left=60] (21) to (40);
      \end{pgfonlayer}
    \end{tikzpicture}
  }
  \vspace*{-3em}
  \caption{Partitioning a sequence of intervals sorted by start time using
    \textsc{NextSubseq} with parameter $b=9$, with the start time being the \emph{key function}.
    \label{fig:example-partitioning}
  }
  \Description{}
\end{figure*}

\begin{algorithm}[t]
  \small
  \caption{\nextsubseq$(\R, h, b, k)$\label{alg:next-breakpoint}}
  \KwIn{Relation $\R$ sorted according to $k$; current index position
    $h$; subsequence size $b$; \emph{key} function $k$}%
  \KwOut{Subsequence of \R starting from position $h$, either of size
    $\le b$ or with all tuples having the same key}

  \BlankLine %
  \If{$h + b \ge |\R|$}{%
    \Return subsequence $\langle\R_h, \dots, \R_{|\R|-1}\rangle$\;%
  }%

  \BlankLine%
  \Let{$h'$}{$h+b$}\;%

  \uIf{$k(\R_{h}) = k(\R_{h'})$} {%
    \nllabel{line:heavy-bucket}%
    \While{$h' < |\R| \wedge k(\R_{h}) = k(\R_{h'})$} {%
      \Let{$h'$}{$h'+1$}%
      \nllabel{line:heavy-bucket-end}%
    }%
  } \Else {%
    \nllabel{line:split-bucket}%
    \While{$k(\R_{h'-1}) = k(\R_{h'})$} {%
      \Let{$h'$}{$h'-1$}%
      \nllabel{line:split-bucket-end}%
    }%
  }%
  \Return subsequence $\langle \R_h,\R_{h+1},\dots,\R_{h'-1}\rangle$\;
\end{algorithm}

\begin{algorithm}[t]
	\small
	\caption{\textsc{BuildIndex}$(\R, s)$\label{alg:build-index}}
	\KwIn{Temporal relation $\R$ and page size parameter $s$}
        \KwOut{Grid $G$ partitioning $\R$ by start time and duration,
          along with auxiliary arrays.}

	\BlankLine
	\Let{\Grid}{[][]}\;
	\Let{\Colbounds}{[]}\;
	\Let{\LatestEnds}{[]}\;
	\Let{\Cellbounds}{[][]}\;
	\Let{\MaxDurations}{[][]}\;

	\BlankLine
	Sort $\R$ by interval start time\;
	\Let{$h$}{$0$} \tcc{position in \R}
	\Let{$i$}{$0$} \tcc{column index}
	\While{$h < |\R|$}{
		\Let{\Column}{\algo{NextSubseq}$(\R, h, s^2, \rT \rightarrow \rTS)$}\;
		\nllabel{line:column-break}
		\Let{\Colboundsi}{$\min\{\rTS : r \in \Column\}$}\;
		\Let{\LatestEndsi}{
			$\max\{\LatestEndsPrev, \{\rTE : r \in \Column\}\}$
		}\;

		\BlankLine
		Sort \Column by duration\;

		\Let{$k$}{$0$} \tcc{position in \Column}
		\Let{$j$}{$0$} \tcc{cell index}
		\While{$k < |\Column|$}{
			\Let{\Cell}{\algo{NextSubseq}$(\Column, k, s, \rT \rightarrow |\rT|)$}\;
			\nllabel{line:cell-break}
			Sort \Cell by end time\;
			\Let{\Gridij}{\Cell}\;
			\Let{\Cellboundsij}{$~\qquad\qquad\qquad\min\{|\rT| : r \in \Cell\}$}\;
			\Let{\MaxDurationsij}{$~\qquad\qquad\qquad\max\{|\rT| : r \in \Cell\}$}\;
			\Let{$j$}{$j+1$}\;
			\Let{$k$}{$k + |\Cell|$}\;
		}
		\Let{$i$}{$i + 1$}\;
		\Let{$h$}{$h + |\Column|$}\;
	}

	\Return (\Grid, \Colbounds, \LatestEnds, \Cellbounds, \MaxDurations)\;
\end{algorithm}

\begin{example}
  Figure~\ref{fig:example-partitioning} depicts three invocations of
  \nextsubseq on a sequence of sorted start times, with parameter
  $b=9$.  The first jump by 9 positions would split the run of
  intervals with start time $13$.  Therefore, the algorithm iterates
  back until the first start time $< 13$.  The second invocation would
  again split the same run since it contains more than 9 intervals
  with start time $13$.  This time, since the endpoints of the jump
  have the same value, the algorithm iterates forward until the last
  interval with the same start time, thus finding a \emph{heavy}
  subsequence.  The last jump defines the third group.
\end{example}

\begin{figure*}
  \includegraphics[width=\textwidth]{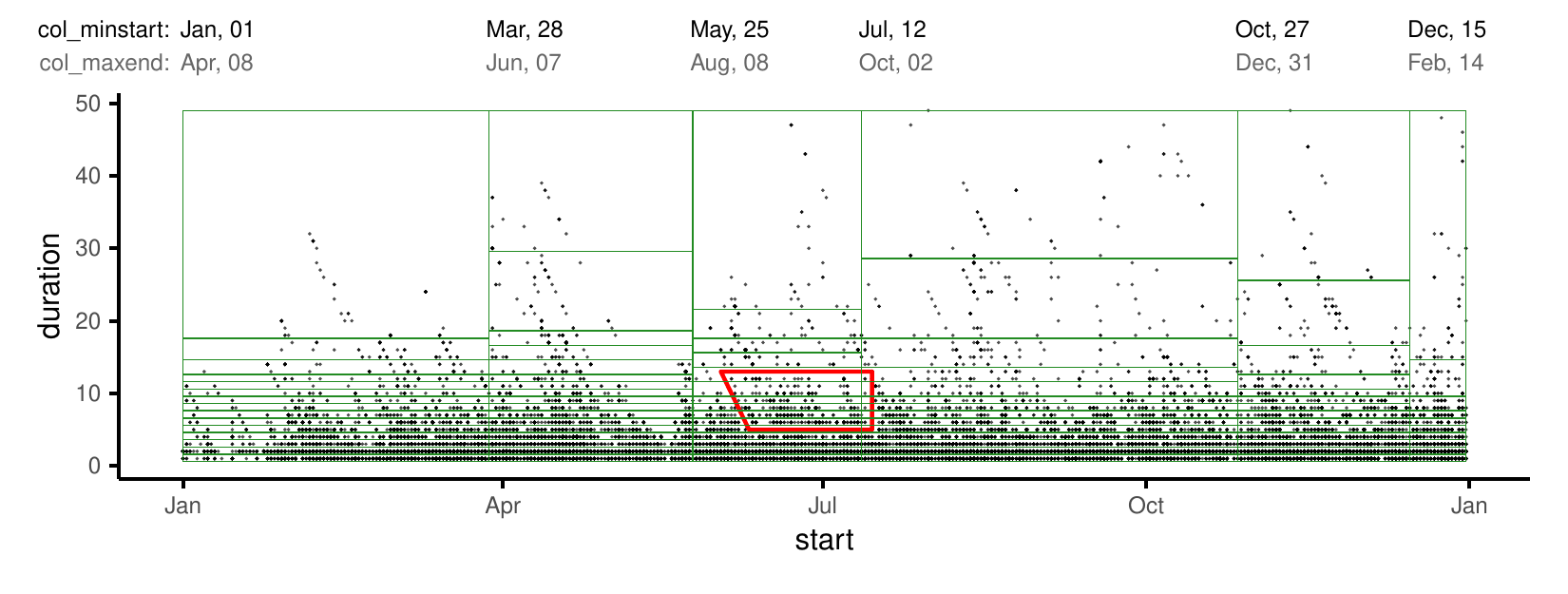}
  \caption{Instantiation of \rd with page size $s=70$ on the example dataset. 
    Above the plot we report the \Colbounds and \LatestEnds arrays.
    \label{fig:example-grid}}
\end{figure*}
    
We are now ready to describe the index construction procedure
\buildindex, which is shown in Algorithm~\ref{alg:build-index}. Let
$s$ be the \emph{page size} parameter, and \R be the relation to be
indexed. First, we sort \R by increasing start time.
Then, we repeatedly invoke \nextsubseq (line~\ref{line:column-break})
to divide the relation into \emph{columns} of $s^2$ tuples each, based
on their start times.  Defining columns in this way allows,
intuitively, to then further partition each column in $s$ cells of $s$
elements each\footnote{We could use two separate parameters to control
  the number of partitions in the two dimensions, but using $s$ and
  $s^2$ makes for a simpler exposition, analysis and
  implementation.}. As a special case, if there are more than $s^2$
tuples with the same start time, \textsc{NextSubseq} will assign them
to the same column, which we then deem \emph{heavy}.  We also keep
track of the minimum start time in each column using an auxiliary
array \Colbounds.  Similarly, the array \LatestEnds stores the
cumulative maximum end time in the columns, i.e., the maximum end time
found so far in the relation.  Both of these arrays will be used at
query time: the first to find the first column to inspect for a given
query, the second to determine when to stop iterating through columns.
Note that \LatestEnds records the \emph{cumulative} maximum end time
of columns.  This ensures that no interval in columns $\le i$ ends
after \LatestEndsi, which will be useful at query time.

Each column is further partitioned in \emph{cells} in a similar
way. First we sort the tuples in the column by increasing
duration. Then we define cells of size $s$ by repeatedly invoking
\textsc{NextSubseq} (line~\ref{line:cell-break}).  Again, if there are
more than $s$ tuples with the same duration, they will all be assigned
to the same cell, which will then be deemed \emph{heavy}.  Similarly
to columns, also cells are complemented by two arrays of ancillary
information: \Cellbounds stores the minimum duration in each cell to
be used at query time to find the first cell to inspect, while
\MaxDurations stores the maximum duration in each cell, which at query
time will determine when to stop iterating through cells.

Finally, tuples in each cell are sorted by the end time.
This is useful at query time, since it allows to stop queries early, as we shall discuss in the proof of Theorem~\ref{thm:query}.

We now formally define light and heavy columns and cells, since they play a key role in the proof of the performance of our index structure.

\begin{definition}\label{def:light-heavy}
  For a given page size $s$,
  a \emph{heavy} column (resp. cell) contains $> s^2$ (resp. $> s$) intervals.
  Conversely, a \emph{light} column (resp. cell) contains $\le s^2$ (resp. $\le s$) intervals.
\end{definition}

\begin{example}
  Figure~\ref{fig:example-grid} shows the grid constructed by
  Algorithm~\ref{alg:build-index} on our example relation from
  Figure~\ref{fig:example-plane}, with page size parameter
  $s=70$. Note that the columns, which contain $s^2=4\,900$ tuples
  each, span different ranges of start times, adapting to the density
  of the points.  Within each column, the points are partitioned
  according to the distribution of durations.  Since many drug
  prescriptions have the same short durations, the cells at the bottom
  of the columns are \emph{heavy} (or span only a few different
  duration values).
  In this setting, a uniform grid
  would suffer from a high imbalance in the number of tuples in each
  cell.
\end{example}

\subsection{Querying the Index}

Given a range-duration query with time range $t = [\TS, \TE)$ and
duration interval $d = [\dmin,\dmax]$, recall that a tuple $r \in \R$
satisfies the query if the
conditions~\eqref{eq:duration-condition},~\eqref{eq:start-condition},
and~\eqref{eq:end-condition} specified in Section~\ref{sec:overview} are
met.  We report the conditions here for convenience:
\begin{align}
  & |\rT| \in [\dmin, \dmax] \tag{\ref{eq:duration-condition}} \\
  & \rTE < \TS                  \tag{\ref{eq:start-condition}}     \\
  & \rTS < \TE                  \tag{\ref{eq:end-condition}}
\end{align}

The pseudocode for querying the index structure is reported in
Algorithm~\ref{alg:query}.  First, we seek the index of the last
column that might contain matching tuples.  To this end, we perform a
binary search on the array \Colbounds to find the last column $i$ such
that the minimum start time in the column (which is the column bound)
is strictly less than the query end time \TE \
(line~\ref{line:search-column}).  This ensures that the column contains at least one tuple
satisfying condition \eqref{eq:end-condition}.  Then, we iterate
\emph{backwards} through columns until column $i$ cannot possibly
contain tuples satisfying the query.  For this, we use the support
array \LatestEnds: if $\LatestEndsi \le \TS$ then we know that all the
columns at index $\le i$ contain tuples that stop earlier than the
start of the query range.  Hence, we can avoid inspecting them because
of condition~\eqref{eq:start-condition}.

For each column that we consider, a binary search on the array
\Cellboundsi is performed, looking for the last cell $j$ such that the
minimum duration in the cell (which is the cell bound) is $\le$ to the
maximum duration $\dmax$ specified in the query
(line~\ref{line:search-cell}). Doing so ensures that at least one
tuple in the cell satisfies the upper bound of condition
\eqref{eq:duration-condition}. Then, we iterate backwards through the
cells until we reach a cell whose maximum duration is less than the
minimum duration $\dmin$ of the query. At this point we stop since
condition~\eqref{eq:duration-condition} can no longer be satisfied.

Finally, we iterate through the tuples of each considered cell by
decreasing end time and stop as soon as
condition~\eqref{eq:start-condition} is no longer satisfied
(line~\ref{line:inner-break}). All the intervals that satisfy the query
are returned in the result.

\begin{algorithm}
  \small
  \caption{\query({\small\Grid, $[\TS, \TE)$, $[\dmin, \dmax]$})\label{alg:query}} 
  \KwIn{
    A range duration query with time range $[\TS, \TE)$ and duration range $[\dmin, \dmax]$; 
    An index \Grid with the ancillary arrays \Colbounds, \Cellbounds, \LatestEnds, and \MaxDurations}

  \BlankLine
  \Let{$res$}{$\emptyset$}\;
  \Let{$i$}{$\argmax_{i} \Colboundsi < \TE$}\;
  \nllabel{line:search-column}
  \While{$i \ge 0 \wedge \TS < \LatestEndsi$} {
    \Let{$j$}{$\argmax_{j} \Cellboundsij \le \dmax$}\;
    \nllabel{line:search-cell}
    \While{$j \ge 0 \wedge \MaxDurationsij \ge \dmin$}{
      \For{$r \in \Gridij$}{
        \If{$\rTE \le \TS$}{
        \nllabel{line:inner-break}
          \textbf{break}\;
        }
        \If{$|\rT| \in [\dmin, \dmax] \wedge \rTS < \TE$}{
          \Let{$res$}{$res \cup \{r\}$}\;
        }
      }
      \Let{$j$}{$j-1$}\;
    }
    \Let{$i$}{$i-1$}\;
  }
  \Return $res$
\end{algorithm}

\subsection{Updating the Index}
\label{sec:updating-index}

Our index data structure can be extended to support both insertion and
removal of tuples.  

\paragraph{Interval Insertion.}
To insert a tuple $r$, we query the index for the start time \rTS and
the duration $|\rT|$ to identify the cell that should contain $r$.
Inserting new intervals into cells might make them grow too large to
be able to maintain the performance guarantees on the query time.
Luckily, as we shall see in Section~\ref{sec:analysis}, \emph{heavy}
columns and cells do not present issues upon insertions by virtue of
containing intervals that all share either the same start time or the
same duration.
If a light column exceeds size $s^2$, we replace it with two new
columns. Similarly, if a light cell exceeds size $s$, we replace it
with two new cells. 

Splitting a column entails to consider all the intervals it contains,
using \nextsubseq with $b = s^2/2 + 1$ to find the breakpoint at which to
split (this way, we balance the size of the new columns). For each of
the two new columns that replace the original column, we apply
\nextsubseq to split them into cells, exactly as in the inner loop of
the index construction.  The array \Colbounds and \LatestEnds are
updated to reflect the replacement of the old column with the new
ones.

Similarly, to split a cell that exceeds size $s$ in column $i$ we use
\nextsubseq with $b=s/2 + 1$ to find a new breakpoint and replace the cell with two new
cells. The auxiliary structures \Cellboundsi and \MaxDurationsi are
updated accordingly.

Both in the case of column and cell splitting, we sort all the intervals in the newly created cells by end time.

\paragraph{Interval Removal.}
As for the removal of a tuple $r$ from the index, we query the index
to find the cell $\Gridij$ that contains $\rT$ and remove the interval
from the cell.  As a consequence, the cell might contain fewer than
$s/2$ items.  
As we shall see with Lemma~\ref{lem:num}, it is crucial for the performance of the index that cells contain at least $s/2$ intervals.

Therefore, upon
removal of an element from a cell, we check whether the sum of
elements of the cell and either of the adjacent ones is less than $s$.
In such case, we \emph{merge} the two cells, i.e., we replace them
with a single cell where all intervals are then sorted by decreasing
end time.  After cells are merged, the arrays \Cellboundsi and
\MaxDurationsi are updated as well to reflect the changes.

Similarly, a removal might cause a column to have fewer than $s^2/2$
elements.  We then apply a similar reasoning.  If the sum of the
number of items in the column and either adjacent ones is smaller than
$s^2/2$, we merge the two columns, i.e., the two columns are replaced by
a single one.  \nextsubseq is then called to find the breakpoints to
divide the newly created column into cells.  After the two columns are
merged, the arrays \Colbounds and \LatestEnds are updated to reflect
the changes.

\section{Analysis}
\label{sec:analysis}

In this section, we provide guarantees on the time required by all the
operations supported by our index structure.  We assume that the index
is built by partitioning first in the time dimension, and then in the
duration dimension.  The asymptotic results presented in this section
hold for both dimension orderings.

\subsection{Querying the Index}

Before stating our main result, we establish the following fundamental
facts about heavy columns and cells.

\begin{lemma}\label{lem:heavy}
  A heavy cell in a heavy column contains only copies of the same interval.
\end{lemma}
\begin{proof}
  By construction, a heavy column contains only intervals with the
  same start time.  Similarly, a heavy cell contains only intervals
  with the same duration.  Therefore, a heavy cell in a heavy column
  contains intervals with the same start time and duration, i.e.,
  multiple copies of the same interval. \qed
\end{proof}

\begin{lemma}\label{lem:num}
  The \rd with parameter $s$ over $n$ intervals has $\BO{\frac{n}{s^2}}$ columns, each
  having $\BO{\frac{n}{s}}$ cells.
\end{lemma}
\begin{proof}
  To upper bound the number of columns into which a relation can be
  partitioned, we first devise a set of intervals that will force \rd
  to use the maximum number of columns.  Consider a relation such that
  exactly $s^2/2+1$ intervals have start time 1, $s^2/2$ intervals
  have start time 2, $s^2/2 + 1$ intervals have start time 3, and so
  on.  \rd will have to build a separate light column for each
  distinct start time.  Note that if the subsets of intervals with the
  same start time were any smaller, \rd would create columns
  containing more than one start time, resulting in fewer columns
  overall.  Therefore, each column has size $\ge s^2/2$, which
  implies that there are $\BO{n/s^2}$ columns.

  As for cells, the worst case occurs if all intervals fall in the same
  column.  With a reasoning similar to the above argument, we have
  that these $n$
  intervals in one column are
  partitioned in $\BO{n/s}$ cells. \qed
\end{proof}

\begin{theorem}\label{thm:query}
  Given an index over a set of $n$ intervals and a page size $s$, the
  time for answering a range-duration query is
  \begin{align*}
    \BO{\frac{n}{s^2}\log\frac{n}{s} + \frac{n}{s} + s^2 + k}
  \end{align*}
  where $k$ is the number of intervals matching the query predicate.
\end{theorem}
\begin{proof}
  Let $\T = [\TS, \TE)$ be the query interval and $d=[\dmin, \dmax]$
  be the duration range of the query.  Recall that the index
  construction algorithm may create a heavy column or cell when it
  cannot break down a group of intervals because they all share the
  same start time or duration.  Furthermore, recall that by
  construction \emph{light} columns (resp. cells) are of size
  $\le s^2$ (resp. $\le s$).

  A range-duration query defines a \emph{query polygon} in the start
  time $\times$ duration space, as depicted in
  Figure~\ref{fig:example-grid}.  For a given cell $\Gridij$, we
  denote by $k_{ij}$ the number of intervals in the cell that match
  the query.
  The query time is comprised of two parts: (a) we have to find the
  correct range of cells to query, and (b) we have to filter the
  intervals in each cell to retrieve the ones that are part of the
  answer.

  As for part (a), finding the cells to inspect, we first determine
  the last column that can contain tuples satisfying the query
  constraint. For this we use binary search over the $\BO{n/s^2}$
  boundaries between columns (Lemma~\ref{lem:num}) in
  $\BO{\log (n/s^2)}$ time.  Iterating backwards through the columns
  visits at most $\BO{n / s^2}$ columns.  In each such column, to find
  the last cell with intervals satisfying the duration constraint, we
  do a binary search over the $\BO{n/s}$ cells in the column
  (Lemma~\ref{lem:num}).  Therefore, the overall time to find the
  correct cells is
  \begin{align*}
    \BO{\log\frac{n}{s^2} + \frac{n}{s^2} \log\frac{n}{s}} =
    \BO{\frac{n}{s^2} \log\frac{n}{s}}
  \end{align*}

  As for part (b), enumerating the results from the cells, we can
  distinguish four possible cases for a cell $\Gridij$:
  \begin{enumerate}[leftmargin=1.3em]
  \item Both the cell's durations range and start times range fall
    within the query polygon. In such cells, the query constraints on
    the duration and the start time are satisfied by construction. The
    algorithm considers just the intervals of $\Gridij$, which are
    part of the output (since it iterates over the cell's intervals by
    decreasing end time). Therefore, the time spent enumerating
    intervals from the cell is $\BO{k_{ij}}$.  Also, note that this is
    the only case in which a heavy cell in a heavy column is visited
    by the query algorithm. Such cells can potentially have $\BO{n}$
    copies of the same interval. If the algorithm visits the cell,
    then it means that such copies are all part of the output, hence
    we pay $\BT{k_{ij}}$ in the complexity.

  \item 
    The cell's start times all satisfy the query constraints, but
    not all its durations do so.
    The cell must be light, otherwise all its durations would be equal and would need to satisfy the duration constraint, contradicting the assumption.
    Therefore, the algorithm has to evaluate $\BO{s}$ intervals.

  \item 
    The cell's durations all satisfy the query constraints, but not all its start times do.
    In this case we cannot be in a heavy column (which contains a single start time) and thus the column contains at most $s^2$ intervals.

  \item Only some of the cell's intervals satisfy both constraints. In
    this case neither the column nor the cell can be heavy, therefore
    the algorithm must only evaluate $\BO{s}$ intervals.
  \end{enumerate}
  All the cells of type (1) account for at most $\BO{k}$ in terms of
  running time.  As for cells of type (2) there are at most
  $\BO{n/s^2}$ of them. Overall they account for $\BO{\frac{n}{s}}$
  interval evaluations.  Cells of type (3) account for $\BO{s^2}$
  interval evaluations overall.  Finally, there is a constant number
  of cells of type (4) (which are at the corners of the query
  polygon).  These cells therefore account for at most $\BO{s}$
  interval evaluations in total.  Overall, the time to evaluate the
  intervals in the cells is $ \BO{\frac{n}{s} + s^2 + k} $.

  Combining this last result with the time to find the range of cells
  to query, we have that the overall query time is
  \begin{align*}
    \BO{\frac{n}{s^2}\log\frac{n}{s} + \frac{n}{s} + s^2 + k}.
  \end{align*}
   \hfill \qed
\end{proof}

The above theorem exposes a fundamental tradeoff of our data
structure: using a smaller page size allows to improve the precision
of the data structure (by looking at fewer intervals that are not part
of the query output), while at the same time increasing the number of
columns that need to be queried.  In Section~\ref{sec:param-dep} we
investigate experimentally the effect of $s$ on the performance,
finding that the best performance is attained for
$n/s^2 \in [50,500]$.

The choice of the order of partitioning has no impact on the
theoretical complexity results, however it affects the practical
performance as we will discuss in Section~\ref{sec:experiments}
(Figure~\ref{fig:param-dependency}). It turns out that generally it is
better to index first the duration and then the start time, as
summarized in the following observation.

\begin{observation}\label{obs:dimension-order}
  While changing the order in which dimensions are indexed does not
  change the asymptotic behavior of \ours, the choice of such order
  affects the practical performance.  In fact range and duration
  constraints play different roles in queries. When querying for
  duration, the position of start times on the timeline does not
  affect the outcome of a query. Conversely, a range query can be
  rephrased in terms of start times and durations.  Therefore, range
  queries can benefit from a partition of the duration dimension,
  since it implies a partition of the end times.  Furthermore,
  duration constraints are typically much more selective than range
  constraints. Therefore, indexing first by duration might imply that
  a query has to iterate through fewer columns, compared to the case
  where start times are indexed first.
\end{observation}

\subsection{Index Construction and Update}

\begin{theorem}
	Given a set of $n$ intervals, building the index requires time $\BO{n \log n}$.
\end{theorem}
\begin{proof}
  Computing column boundaries for the duration requires
  time $\BO{n\log n}$ for sorting the tuples, while invoking
  Algorithm~\ref{alg:next-breakpoint} requires time 
  $\BO{n}$ for iterating through the sorted intervals to find the boundaries: in fact, each of the $n$ intervals is
  visited at most once.

  Let $n_i$ be the number of intervals in the $i$-th column.  Sorting
  the column requires time $\BO{n_i\log n_i}$.  Then, the invocation
  of Algorithm~\ref{alg:next-breakpoint} to find cell boundaries
  requires time $\BO{n_i}$, following the same argument as above.
  Similarly, let $n_{ij}$ be the size of cell $ij$. Sorting its tuples
  by end time requires time $\BO{n_{ij} \log n_{ij}}$.  Therefore, the
  time for running the inner loop is $\BO{n \log n}$ overall, and the
  theorem follows. \qed
\end{proof}

\begin{theorem}
	Inserting a tuple $r$ into the index requires
	time $\BO{\log n/s + n/s + s^2 \log s}$.
\end{theorem}
\begin{proof}
  Finding the cell that should contain $\rT$ requires performing first a binary search over the $\BO{\frac{n}{s^2}}$ columns and then over the $\BO{\frac{n}{s}}$ cells of the resulting column (by Lemma~\ref{lem:num}), for a total time $\BO{\log n/s}$.

  Note that a heavy column or cell does not need to be split, since by
  definition they can contain more than $s^2$ (resp. $s$) intervals,
  as long as they all have the same start time (resp. duration).
  Therefore we consider only the case of splitting a light column
  or light cell.
  Splitting a column with $s^2 + 1$ intervals requires to first apply
  Algorithm~\ref{alg:next-breakpoint} (which takes time $\BO{s^2}$),
  sort the newly created columns (requiring time at most
  $\BO{s^2 \log s}$) and then applying again
  Algorithm~\ref{alg:next-breakpoint} on the new columns to find the
  cells (for an overall time of $\BO{s^2}$).  Furthermore, inserting
  the new column boundaries in order in the sorted array of boundaries
  requires $\BO{\frac{n}{s^2}}$ time.  By similar reasoning, splitting
  a cell with $s+1$ intervals requires time $\BO{s\log s}$, and
  updating cell boundaries by inserting in-order requires
  $\BO{\frac{n}{s}}$. \qed
\end{proof}

In a similar way we can prove the following theorem.

\begin{theorem}
  Removing a tuple $r$ from the index requires time
  $\BO{\log n/s + n/s}$.
\end{theorem}

\section{Experimental Evaluation}
\label{sec:experiments}

In this section, we evaluate experimentally our proposed index.  To
frame this evaluation, we consider the following data structures as
baselines: The implementation of \btree provided by the Rust standard
library, which is optimized for CPU cache usage; intervals are indexed
by duration in this case.  The \itree
index~\cite{DBLP:conf/vldb/KriegelPS00}, which we implemented
ourselves, indexing intervals by their position on the timeline.  The
\gfile~\cite{DBLP:journals/tods/NievergeltHS84} and
\pindex{}~\cite{BehrendDGSVRK_SSTD19_period-index}, which we also
implemented, and which index both start times and durations.  We also
consider the \textsc{R-Tree} as a baseline, specifically the
\rtree~\cite{DBLP:conf/sigmod/BeckmannKSS90}: intervals
are mapped to points identified by the start time and duration of the
interval, and then indexed by the \rtree.  Our
proposed data structure will be denoted with \rdtd when start time is
indexed before duration, and with \rddt otherwise.  In cases where the
order of the dimensions is not relevant to the discussion, we will use
\ours instead.  To account for the potential shortcomings of our
implementations, we will also evaluate the relative performance of the
data structures with implementation-independent
metrics~\cite{DBLP:journals/kais/KriegelSZ17}.

We aim to answer the following questions:
\begin{description}
  \item[(Q1)] How robust are the indices over different workloads? (\textsection~\ref{sec:overview-robust})
  \item[(Q2)] How do indices compare on mixed workloads? (\textsection~\ref{sec:mixed})
  \item[(Q3)] How do query times relate with the query selectivity? (\textsection~\ref{sec:query-inspect})
  \item[(Q4)] How does the page size parameter affect the performance of the index? (\textsection~\ref{sec:param-dep})
  \item[(Q5)] How does the performance scale with respect to the size of the dataset? (\textsection~\ref{sec:scalability})
  \item[(Q6)] What is the performance of updating the index?
    (\textsection~\ref{sec:insertion})
\end{description}

\subsection{Setup and Datasets}

\begin{figure*}
  \includegraphics[width=\textwidth]{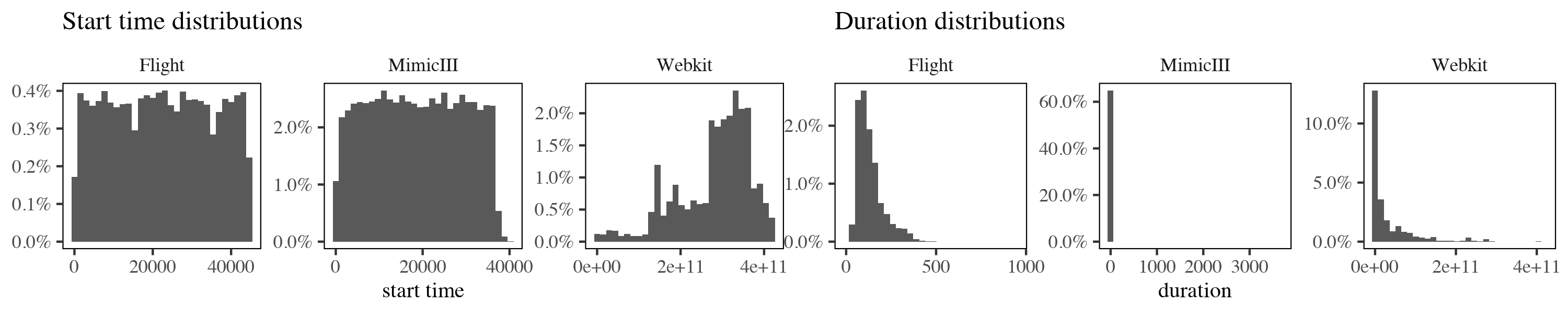}
  \caption{Histograms of the start times and durations of the three
    real-world datasets considered in this
    paper.\label{fig:data-distributions}}
\end{figure*}

We implemented our index and the baseline competitors in Rust 1.44.1,
using a configurable and extensible
framework~\cite{DBLP:conf/sisap/AumullerC20}.  Code and data are
available at~\url{https://github.com/Cecca/temporal-index}.  The
experiments presented in this section were run on a machine equipped
with 94 GB of memory and a Intel\textregistered Xeon\textregistered
CPU E5-2667 v3 @ 3.20GHz processor.

As a benchmark we consider the following datasets and workloads:

\paragraph{\dataset{Flight}:}
A set of 701\,353 flights, identified by their takeoff and landing
time at the granularity of minute, covering August 2018.  Query ranges
on this dataset are generated at random.  Time range durations are
uniformly distributed between one and 31 days, and duration ranges are
uniformly distributed between one minute and one day.

\paragraph{\dataset{Webkit}:}
1\,547\,419 file edits in the Webkit source code repository.
Intervals represent the timespan between successive edits to a file.
Query ranges on this dataset are generated at random.  Time range
durations are uniformly distributed between one minute and one year,
and duration ranges are uniformly distributed between up to three
years.

\paragraph{\dataset{MimicIII}:}
4\,134\,909 drug prescriptions from the open MimicIII
database~\cite{mimiciii}.  Each prescription is characterized by its
start and end day.  Queries, generated at random, span the entire
domain of times and durations: the former take values
$\in [1, 40\,251]$, the latter $\in [1, 200]$.  The large span of
start times (110 years) is due to the anonymization procedure applied
to the database.

\paragraph{\dataset{Synthetic}:} Randomly generated datasets with 10
million intervals by default.  The interval start times are uniformly
distributed in $[1, n]$, where $n$ is the size of the dataset;
interval durations follow a Zipf distribution with $\beta=1$.

\bigskip

Figure~\ref{fig:data-distributions} reports the distribution of start
times and durations of the three real-world datasets.  Note that the
start times of the \dataset{Flight} dataset are approximately
uniformly distributed, with durations geometrically distributed.  The
\dataset{MimicIII} dataset has similarly distributed start times, and
duration values which are very concentrated.  Finally, the
\dataset{Webkit} dataset has a distribution of durations similar to
\dataset{Flight}, but the distribution of start times is skewed
towards recent times, meaning that more file edits have been done
recently.

In the experiments we run a count query subject to range and duration
constraints.  Evaluating a count query rather than reporting the
matching tuples, allows to measure the time to retrieve the tuples,
rather than the time to print them to screen or write them to a file.

\subsection{Robustness of Index Structures Across Different Workloads}
\label{sec:overview-robust}

\begin{table*}
  \caption{ Performance attained by the tested indices on different
    workloads.  Performance is characterized in terms of three
    numbers: the number of queries per second (larger is better), the
    time to build the index in milliseconds (smaller is better), and
    the index size (smaller is better) expressed as the number of
    bytes used for each indexed interval.  Configurations with a dark
    blue background are the ones with the best throughput for a
    dataset/queryset combination, light blue cells denote the
    second-best data structure.  The fastest index construction times
    are underlined, and the smallest indices are highlighted in
    bold. 
    \label{tbl:overview-qps}}
    \centering
  {
    \renewcommand{\arraystretch}{1.5}
    \resizebox{\textwidth}{!}{%

    \begin{tabular}[t]{llrrrrrrr}
    \toprule
    \multicolumn{2}{c}{ } & \multicolumn{7}{c}{Queries per second $~|~$ Index build time $~|~$ Bytes per interval} \\
    \cmidrule(l{3pt}r{3pt}){3-9}
    dataset & query & \rdtd & \rddt & \gfile & \pindex & \rtree & \itree & \btree\\
    \midrule
     & \qro & \colorbox[HTML]{FFFFFF}{\color{black}415}$~|~$1\,024$~|~$\textbf{16.1} & \colorbox[HTML]{3572B9}{\color{white}602}$~|~$889$~|~$16.1 & \colorbox[HTML]{C0DBEC}{\color{black}468}$~|~$\underline{445}$~|~$27.7 & \colorbox[HTML]{FFFFFF}{\color{black}242}$~|~$2\,392$~|~$57.1 & \colorbox[HTML]{FFFFFF}{\color{black}12}$~|~$3\,358$~|~$80.0 & \colorbox[HTML]{FFFFFF}{\color{black}272}$~|~$3\,636$~|~$50.1 & \colorbox[HTML]{FFFFFF}{\color{black}8}$~|~$2\,505$~|~$31.4\\

     & \qdo & \colorbox[HTML]{C0DBEC}{\color{black}842}$~|~$1\,030$~|~$\textbf{16.1} & \colorbox[HTML]{FFFFFF}{\color{black}734}$~|~$\underline{785}$~|~$17.1 & \colorbox[HTML]{FFFFFF}{\color{black}77}$~|~$983$~|~$24.1 & \colorbox[HTML]{FFFFFF}{\color{black}69}$~|~$6\,139$~|~$178.1 & \colorbox[HTML]{FFFFFF}{\color{black}106}$~|~$3\,398$~|~$80.0 & \colorbox[HTML]{FFFFFF}{\color{black}12}$~|~$3\,636$~|~$50.1 & \colorbox[HTML]{3572B9}{\color{white}43\,103}$~|~$2\,469$~|~$31.4\\

    \multirow{-3}{*}{\raggedright\arraybackslash Synthetic} & \qrd & \colorbox[HTML]{C0DBEC}{\color{black}11\,737}$~|~$1\,052$~|~$\textbf{16.1} & \colorbox[HTML]{3572B9}{\color{white}14\,085}$~|~$\underline{829}$~|~$16.5 & \colorbox[HTML]{FFFFFF}{\color{black}1\,581}$~|~$857$~|~$24.1 & \colorbox[HTML]{FFFFFF}{\color{black}1\,403}$~|~$6\,566$~|~$178.1 & \colorbox[HTML]{FFFFFF}{\color{black}2\,591}$~|~$3\,402$~|~$80.0 & \colorbox[HTML]{FFFFFF}{\color{black}215}$~|~$3\,674$~|~$50.1 & \colorbox[HTML]{FFFFFF}{\color{black}502}$~|~$2\,529$~|~$31.4\\
    \cmidrule{1-9}
     & \qro & \colorbox[HTML]{C0DBEC}{\color{black}54\,945}$~|~$42$~|~$17.1 & \colorbox[HTML]{3572B9}{\color{white}57\,471}$~|~$44$~|~$\textbf{16.1} & \colorbox[HTML]{FFFFFF}{\color{black}49\,505}$~|~$\underline{20}$~|~$23.2 & \colorbox[HTML]{FFFFFF}{\color{black}15\,015}$~|~$80$~|~$40.0 & \colorbox[HTML]{FFFFFF}{\color{black}4\,921}$~|~$195$~|~$97.9 & \colorbox[HTML]{FFFFFF}{\color{black}51\,813}$~|~$139$~|~$44.5 & \colorbox[HTML]{FFFFFF}{\color{black}2\,798}$~|~$39$~|~$23.3\\

     & \qdo & \colorbox[HTML]{FFFFFF}{\color{black}4\,182}$~|~$44$~|~$\textbf{16.1} & \colorbox[HTML]{C0DBEC}{\color{black}4\,218}$~|~$33$~|~$21.3 & \colorbox[HTML]{FFFFFF}{\color{black}3\,791}$~|~$\underline{20}$~|~$23.2 & \colorbox[HTML]{FFFFFF}{\color{black}538}$~|~$111$~|~$62.9 & \colorbox[HTML]{FFFFFF}{\color{black}378}$~|~$205$~|~$97.9 & \colorbox[HTML]{FFFFFF}{\color{black}683}$~|~$150$~|~$44.5 & \colorbox[HTML]{3572B9}{\color{white}714\,286}$~|~$43$~|~$23.3\\

    \multirow{-3}{*}{\raggedright\arraybackslash Flight} & \qrd & \colorbox[HTML]{3572B9}{\color{white}232\,558}$~|~$45$~|~$16.7 & \colorbox[HTML]{C0DBEC}{\color{black}208\,333}$~|~$39$~|~$\textbf{16.6} & \colorbox[HTML]{FFFFFF}{\color{black}188\,679}$~|~$\underline{20}$~|~$23.2 & \colorbox[HTML]{FFFFFF}{\color{black}29\,940}$~|~$95$~|~$62.9 & \colorbox[HTML]{FFFFFF}{\color{black}21\,786}$~|~$193$~|~$97.9 & \colorbox[HTML]{FFFFFF}{\color{black}43\,860}$~|~$139$~|~$44.5 & \colorbox[HTML]{FFFFFF}{\color{black}13\,514}$~|~$47$~|~$23.3\\
    \cmidrule{1-9}
     & \qro & \colorbox[HTML]{FFFFFF}{\color{black}416}$~|~$116$~|~$\textbf{16.1} & \colorbox[HTML]{3572B9}{\color{white}467}$~|~$115$~|~$16.1 & \colorbox[HTML]{C0DBEC}{\color{black}445}$~|~$\underline{56}$~|~$23.0 & \colorbox[HTML]{FFFFFF}{\color{black}83}$~|~$540$~|~$119.0 & \colorbox[HTML]{FFFFFF}{\color{black}30}$~|~$468$~|~$97.7 & \colorbox[HTML]{FFFFFF}{\color{black}384}$~|~$274$~|~$48.3 & \colorbox[HTML]{FFFFFF}{\color{black}47}$~|~$252$~|~$36.4\\

     & \qdo & \colorbox[HTML]{C0DBEC}{\color{black}2\,544}$~|~$112$~|~$\textbf{16.1} & \colorbox[HTML]{FFFFFF}{\color{black}2\,520}$~|~$104$~|~$16.9 & \colorbox[HTML]{FFFFFF}{\color{black}1\,938}$~|~$\underline{52}$~|~$23.0 & \colorbox[HTML]{FFFFFF}{\color{black}57}$~|~$559$~|~$124.4 & \colorbox[HTML]{FFFFFF}{\color{black}234}$~|~$469$~|~$97.7 & \colorbox[HTML]{FFFFFF}{\color{black}157}$~|~$274$~|~$48.3 & \colorbox[HTML]{3572B9}{\color{white}3\,393}$~|~$255$~|~$36.4\\

    \multirow{-3}{*}{\raggedright\arraybackslash Webkit} & \qrd & \colorbox[HTML]{3572B9}{\color{white}3\,159}$~|~$112$~|~$\textbf{16.1} & \colorbox[HTML]{C0DBEC}{\color{black}3\,133}$~|~$99$~|~$16.9 & \colorbox[HTML]{FFFFFF}{\color{black}2\,322}$~|~$\underline{51}$~|~$23.0 & \colorbox[HTML]{FFFFFF}{\color{black}80}$~|~$737$~|~$322.6 & \colorbox[HTML]{FFFFFF}{\color{black}258}$~|~$468$~|~$97.7 & \colorbox[HTML]{FFFFFF}{\color{black}256}$~|~$274$~|~$48.3 & \colorbox[HTML]{FFFFFF}{\color{black}758}$~|~$252$~|~$36.4\\
    \cmidrule{1-9}
     & \qro & \colorbox[HTML]{FFFFFF}{\color{black}372}$~|~$240$~|~$\textbf{16.0} & \colorbox[HTML]{3572B9}{\color{white}391}$~|~$213$~|~$16.0 & \colorbox[HTML]{3572B9}{\color{white}391}$~|~$143$~|~$23.5 & \colorbox[HTML]{C0DBEC}{\color{black}381}$~|~$451$~|~$24.3 & \colorbox[HTML]{FFFFFF}{\color{black}38}$~|~$1\,112$~|~$81.2 & \colorbox[HTML]{FFFFFF}{\color{black}347}$~|~$717$~|~$46.6 & \colorbox[HTML]{FFFFFF}{\color{black}127}$~|~$\underline{138}$~|~$21.2\\

     & \qdo & \colorbox[HTML]{C0DBEC}{\color{black}3\,560}$~|~$275$~|~$\textbf{16.0} & \colorbox[HTML]{FFFFFF}{\color{black}2\,876}$~|~$202$~|~$17.5 & \colorbox[HTML]{FFFFFF}{\color{black}2\,075}$~|~$184$~|~$27.7 & \colorbox[HTML]{FFFFFF}{\color{black}520}$~|~$718$~|~$41.3 & \colorbox[HTML]{FFFFFF}{\color{black}452}$~|~$1\,175$~|~$81.2 & \colorbox[HTML]{FFFFFF}{\color{black}76}$~|~$733$~|~$46.6 & \colorbox[HTML]{3572B9}{\color{white}2\,500\,000}$~|~$\underline{141}$~|~$21.2\\

    \multirow{-3}{*}{\raggedright\arraybackslash MimicIII} & \qrd & \colorbox[HTML]{C0DBEC}{\color{black}10\,246}$~|~$241$~|~$\textbf{16.0} & \colorbox[HTML]{3572B9}{\color{white}10\,730}$~|~$201$~|~$17.5 & \colorbox[HTML]{FFFFFF}{\color{black}5\,227}$~|~$185$~|~$27.7 & \colorbox[HTML]{FFFFFF}{\color{black}1\,779}$~|~$720$~|~$41.3 & \colorbox[HTML]{FFFFFF}{\color{black}1\,384}$~|~$1\,081$~|~$81.2 & \colorbox[HTML]{FFFFFF}{\color{black}232}$~|~$721$~|~$46.6 & \colorbox[HTML]{FFFFFF}{\color{black}3\,347}$~|~$\underline{135}$~|~$21.2\\
    \bottomrule
    \end{tabular}
    }
  }
\end{table*}

In the first set of experiments, we consider both real-world and
synthetic datasets (with 10 million intervals).
Table~\ref{tbl:overview-qps} reports, under different combinations of
dataset/query workload, an overview on the performance of different
index structures on three indicators: the queries per second, the time
to build the index, and the size of the index.  The latter is measured
in terms of bytes per interval, i.e., the number of bytes that the
index uses for each input interval. Since we are representing
intervals as pairs of 64-bits unsigned integers, 16 bytes per interval
are required just to represent the data, and thus are a lower bound on
this performance metric.  
Dark blue
and light blue cells denote, respectively, the best and second-best
performing data structures in terms of queries per second.  For index
structures that take parameters, we report on the best configuration.
In particular, we defer the discussion of the effect of different
parameterizations of \ours to Section~\ref{sec:param-dep}.

We remark that the overall difference in throughput for different
workloads is due to the different output sizes: \qdo queries are in general
less selective than \qrd queries, hence it takes more time to iterate
through the output.  This explains why, in general, \qrd queries enjoy
a higher throughput across all the index structures.

Consider first the performance in terms of throughput, i.e., the first
measure.  We observe that \ours always performs better than
competitors on \qrd queries and on \qro queries (with the exception of
\dataset{MimicIII}, where it ties with \gfile).  For \qdo queries it
is always the second best solution after the \btree.  We will,
however, see in the next section how \ours surpasses the \btree as
soon as a few queries constraining also the range are introduced in
the workload.

We also observe that the \algo{grid-file} ranks second or third in
several cases.  Recall that this data structure is rather similar to
\ours, the difference being that the latter is adaptive to the input
distribution.  This shows the performance benefits of a data structure
that takes into account the data distribution.

Concerning the other performance indicators, we note that the index
construction time of \ours is comparable with the one of the \gfile
and \btree, and generally much faster than the other approaches.  As
for the size of the index, \ours always produces the smallest index,
across all tested configurations, using just slightly more than the
minimum 16 bytes to represent an interval.  The other approaches,
especially pointer-based data structures such as \btree, \itree, and
\rtree require significantly more space. The \pindex has a much higher
space requirement compared to the others, since each interval may be
replicated several times.

In summary, \ours is a data structure that provides fast query times,
is fast to build, and has negligible space overhead.

\subsection{Mixed Query Workloads}
\label{sec:mixed}

\begin{figure}
  \includegraphics[width=\columnwidth]{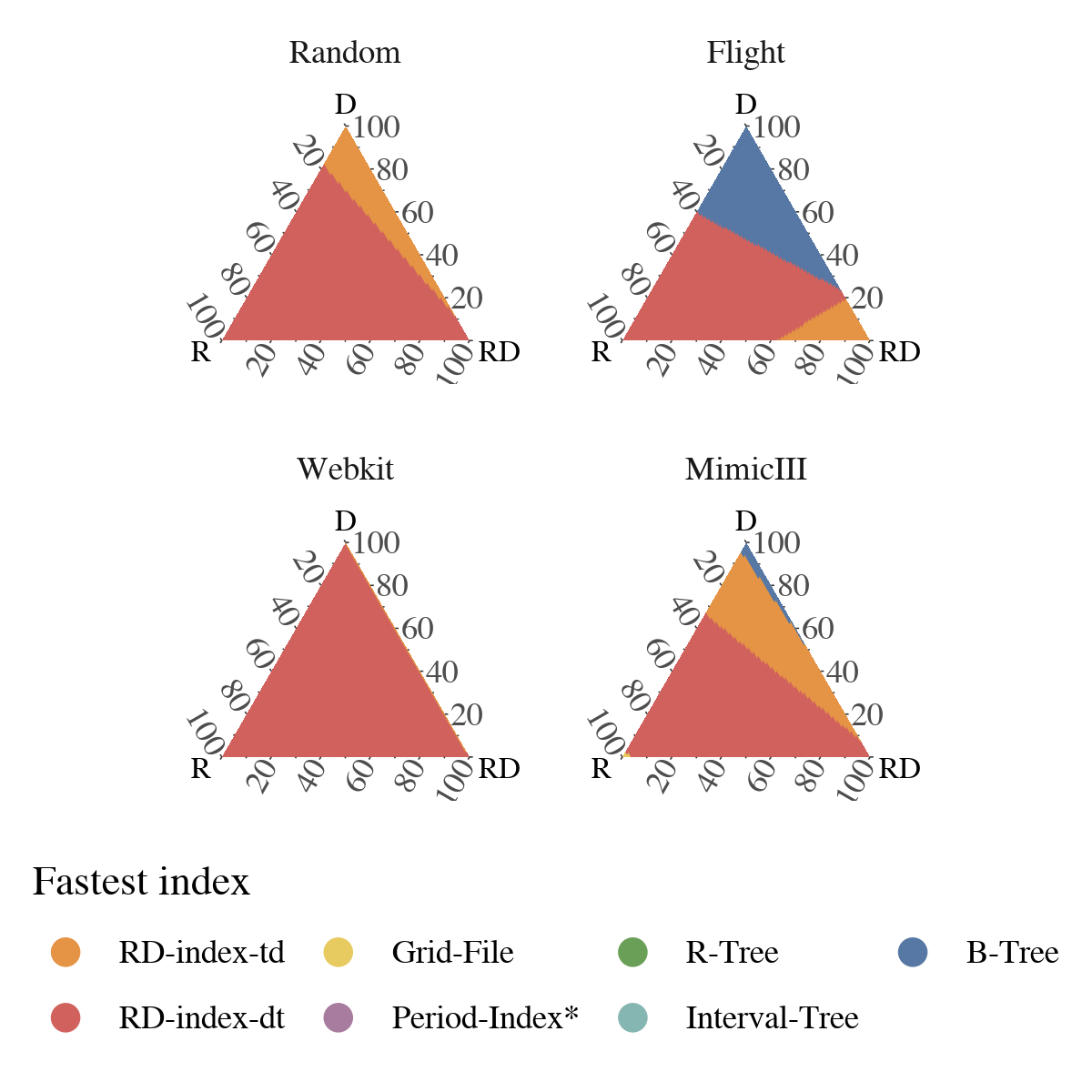}
  \caption{\label{fig:mixed} Ternary plots of best performing index
    structure for different mixed workloads, on the four different
    datasets.  The online supplementary material provides an
    interactive version of this plot.  }
\end{figure}

We now consider mixed query workloads, consisting of a mix of \qro,
\qdo, and \qrd queries.  Rather than fixing a particular combination
of queries, we use the data of Table~\ref{tbl:overview-qps} to
estimate the throughput of workloads composed by any combination of
queries.
Let $f_{rd}$, $f_{r}$, and
$f_{d}$ be, respectively, the fraction of \qrd, \qro, and \qdo queries
in the mixed query workload, with $f_{rd} + f_{r} + f_{d} = 1$.
Similarly, for a given algorithm and dataset, let $\phi_{rd}$,
$\phi_{r}$, and $\phi_{d}$ be the throughputs of \qrd, \qro, and \qdo
queries, as reported in Table~\ref{tbl:overview-qps}.
The throughput
of the combined workload will then be the weighted harmonic mean of the average rates reported in Table~\ref{tbl:overview-qps}
\begin{align*}
  \frac{1}{
	  \frac{f_{rd}}{\phi_{rd}} +
	  \frac{f_{d}}{\phi_{d}} +
	  \frac{f_{r}}{\phi_{r}}
  }.
\end{align*}

Figure~\ref{fig:mixed} provides a summary of the best performing
algorithm for any workload that can be concocted with the formula
above.  In each ternary plot, each point in the triangle identifies a
combination of \qro, \qdo, and \qrd queries.  For instance, the center
of each triangle corresponds to a workload composed in equal parts by
the three types of queries.  Portions of the triangles are colored
according to the best-performing index for the corresponding
workloads.
We observe that \rd is the best performing index structure in the vast
majority of workloads.  The exception is for workloads where \qdo
queries are the majority, where the \btree outperforms \rd.  On mixed
workloads, we observe that the performance of \rd is rather robust to
the ordering of indexing dimensions. This can be verified in
Figure~\ref{fig:tradeoff-all} in
Appendix~\ref{sec:additional-figures}, which reports ternary plots of
the throughput of all algorithms on all datasets for all mixed
workloads.
The online supplemental
material\footnote{\url{https://cecca.github.io/temporal-index/}}
provides an interactive tool to explore the performance of all
different index structures on any mixed workload.

\subsection{Distribution of Query Times Against Selectivity}
\label{sec:query-inspect}

\begin{figure*}
	\includegraphics[width=\textwidth]{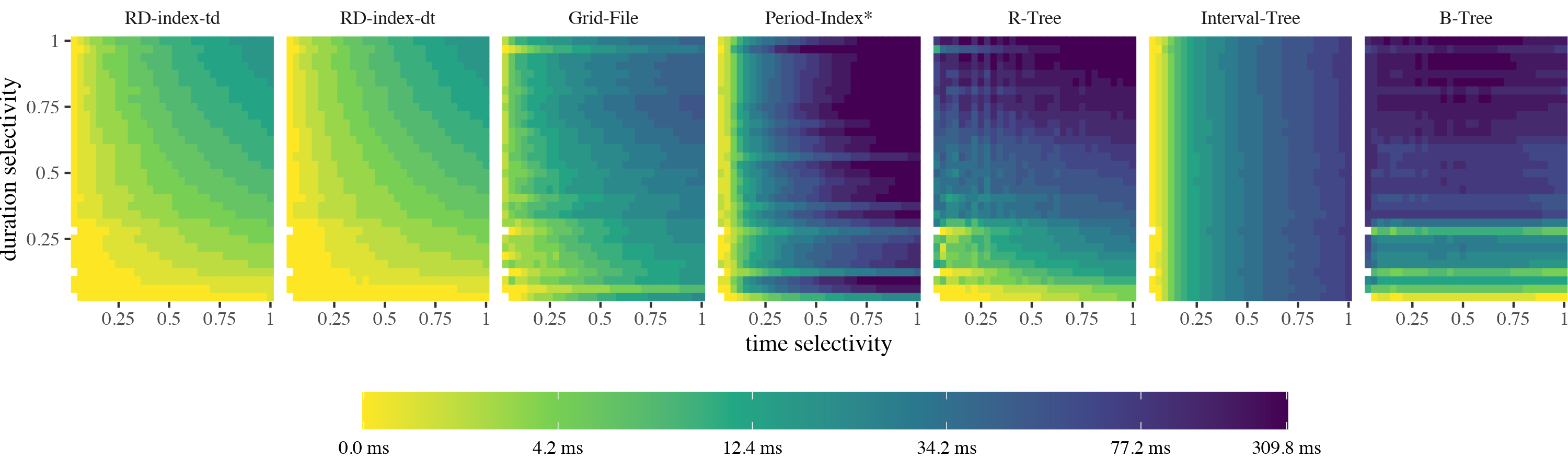}
	\includegraphics[width=\textwidth]{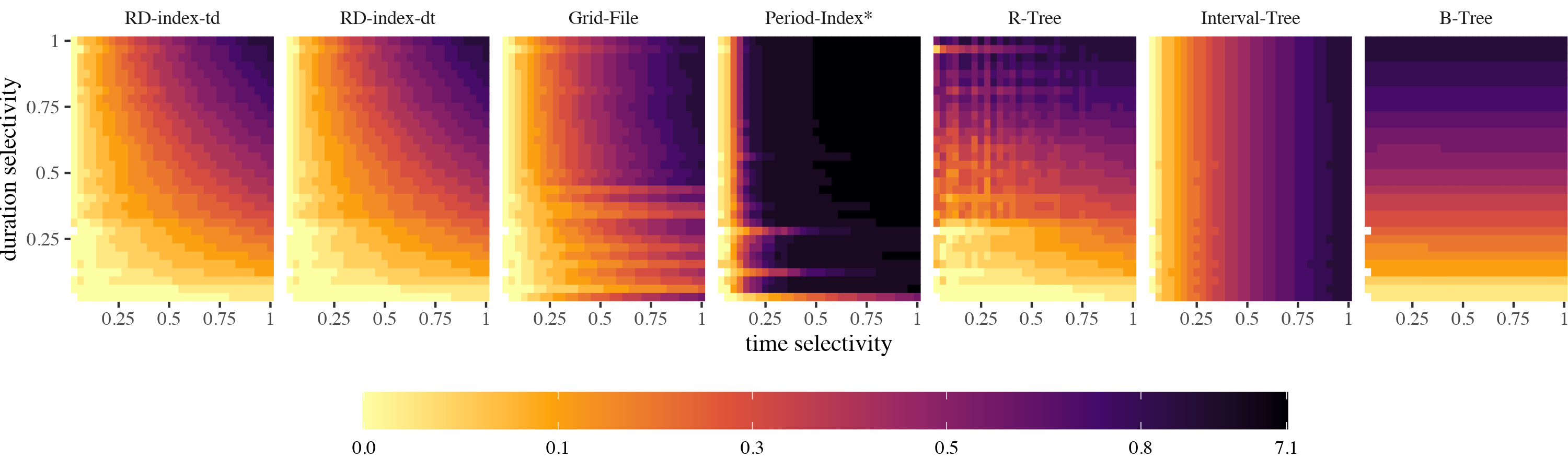}
	\caption{
		Heatmaps of the performance of index structures against selectivity
		in the time and duration dimensions.
		The top row reports the time, in milliseconds, required to answer a query with a given time and duration selectivity.
		The plots report the fraction of intervals (over the total $n=10^7$) examined by each query. A fraction larger than 1 means that
		a query examined the same interval more than once.
		For readability the color scales are binned every
		5 percentiles, therefore the scales are non linear.
		In both plots, lighter is better.
		\label{fig:selectivity-grid}}
	\Description{Time versus selectivity in two dimensions}
\end{figure*}

\begin{figure*}
	\includegraphics[width=\textwidth]{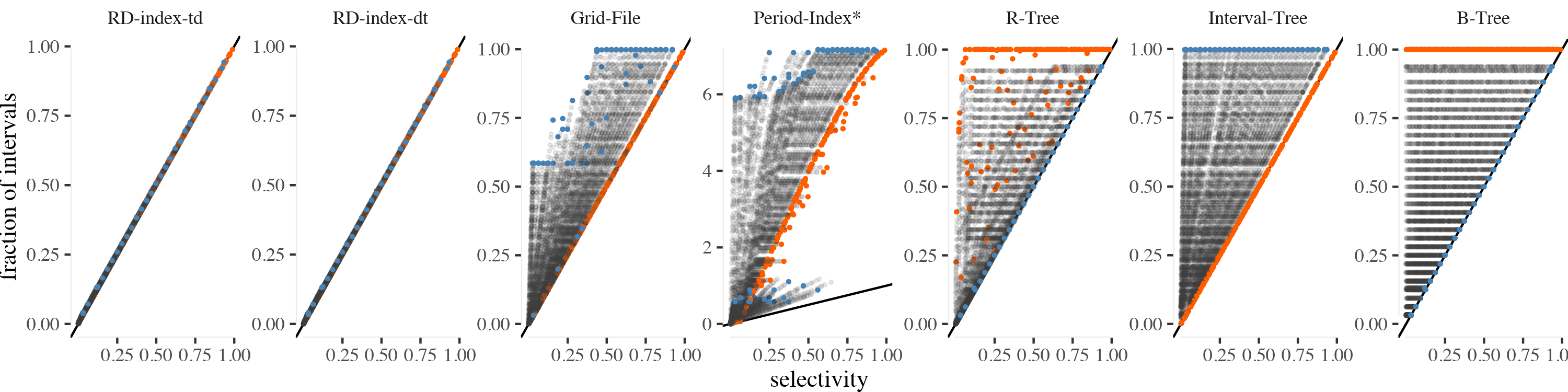}
	\caption{
		Number of examined intervals against selectivity of the query.
		Each point represents a query:
		red dots {\color[HTML]{ff5e00}$\bullet$} are \qro queries, 
		blue dots {\color[HTML]{4682B4}$\bullet$} are \qdo queries,
		black dots $\bullet$ are \qrd queries.
		The black line represents the ideal behavior, in which only matching
		intervals are examined.
		Note that all the plots except for the one related to the 
		\pindex share the $y$ axis: in all plots the black diagonal has slope 1.
		\label{fig:time-vs-selectivity}}
	\Description{Time versus selectivity}
\end{figure*}

We investigate the relationship between the selectivity of queries
(i.e., the fraction of the input that satisfies them) and the time
taken by different data structures to answer them.  Given that
range-duration queries constrain both the time and the duration
ranges, for a query we can define the \emph{time selectivity} as the
fraction of the input that satisfies the time range constraint of the
query, and the \emph{duration selectivity} as the fraction of the
input satisfying the duration constraint of the query.

For a given dataset, we build a query set such that queries are
uniformly distributed on the \emph{time}\,$\times$\,\emph{duration}
selectivity plane: this way we have queries that are very selective in
only one dimension, very selective in both dimensions, or not
selective at all.  The goal is to investigate the behavior of each
index structure for queries with different characteristics.
To account for the overhead of measuring time and to level out the
effect of the CPU cache, we run each query 100 consecutive times and
report the average.
For \ours, we set the page size to $s=200$, which is a good parameter choice of all the datasets we consider.

Figure~\ref{fig:selectivity-grid} reports the results for such a
setup, with 1024 queries arranged in a 32\,$\times$\,32 grid, running
on a synthetic dataset with 10 million intervals, with uniformly
distributed start times and Zipf distributed durations.
Colors encode the fraction of
the dataset inspected by each query.  This metric allows a more
implementation-independent assessment of the relative performance of
different data structures.  When reading
Figure~\ref{fig:selectivity-grid}, remember that less selective
queries require more time just to iterate through the output.
Interestingly, different data structures exhibit different patterns in
this plot, as a result of how they access data.  

The \itree plot exhibits vertical bands. The data structure is able to
select intervals only based on their position on the timeline.
Therefore, for a fixed time selectivity of the query the fraction of
intervals inspected (and thus the time to answer the query) does not
depend on the selectivity in the duration dimension, since all the
candidate intervals need to be examined.  For similar reasons, the
\btree shows horizontal bands.  Data structures that explicitly index
both dimensions, instead, tend to exhibit a more \emph{diagonal}
pattern, in particular \ours{}, with a milder effect for \gfile and
\rtree.  The pattern exhibited by \pindex tends to be more similar to
the \itree.  This means that this index is more responsive to queries
that are selective in the time dimension.  This is to be expected
since \pindex is adaptive to the distribution of start times in the
dataset.  Its worse performance compared to \ours and \gfile for
queries of a given selectivity is explained by the fact that some
intervals might be represented multiple times in the index.

Figure~\ref{fig:time-vs-selectivity} reports the performance against
the \emph{overall} selectivity of the same queries.  The performance
is assessed in terms of the fraction of the input inspected by each
query, which are represented as dots whose position along the $x$ axis
encodes their overall selectivity.  Ideally, a data structure should
answer queries by inspecting just the tuples which are part of the
output.  This behavior is represented by the black diagonal in
Figure~\ref{fig:time-vs-selectivity}.  We observe that \ours indeed
shows the ideal behavior.
As for the \gfile, since the efficiency in answering a query depends
on the density of the cells being considered, the performance is in
many cases far from ideal. This can be seen from the fact that several
queries are far away from the ideal diagonal.
The \btree indexes intervals by their duration. As such, \qdo queries
are answered most efficiently: In Figure~\ref{fig:time-vs-selectivity}
such queries lie on the ideal diagonal.  On the other hand, \qro
queries are answered by simply enumerating the entire dataset, thus
scoring 1 on Figure~\ref{fig:time-vs-selectivity}.  Similar
considerations hold for the \itree, with \qro queries performing the
best and \qdo queries performing the worst.
Finally, \pindex inspects the same tuples multiple times for the
majority of queries, which are thus very far away from the
ideal diagonal line in the plot.

\begin{figure}
  \centering
  \includegraphics[width=0.8\columnwidth]{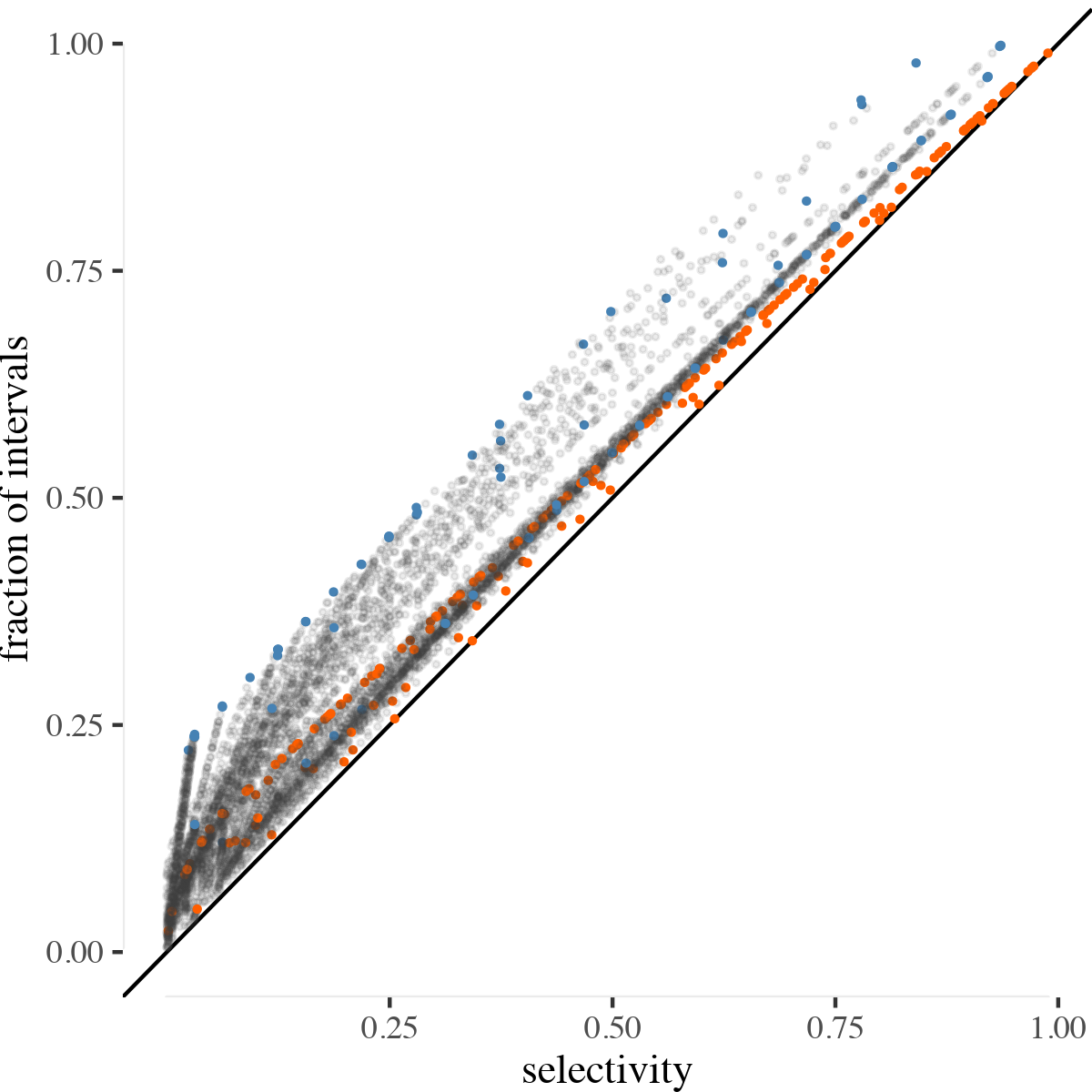}
  \caption{Performance of \ours{} under a non-optimal parameter configuration.
  \label{fig:inefficient-configuration}}
\end{figure}

So far, we have reported on the performance of a good configuration (i.e. $s=200$)  of
\ours{}, which gives a behavior close to optimal.
Figure~\ref{fig:inefficient-configuration} reports, 
in the same scenario
of Figure~\ref{fig:time-vs-selectivity}, the performance of \ours{}
with a time-duration ordering and page size $s=10$.  As we will see in the
next section, this configuration is not ideal (the page size is too small): indeed, for several
queries we have to visit more intervals than the precision of the
query.  This is especially true for duration-only queries (blue dots).
Nevertheless, even under this non-optimal setting, \ours still
performs better than the other approaches, i.e., points are more
concentrated towards the diagonal.

\subsection{Dependency on the Page Size and the Indexing Order}
\label{sec:param-dep}

The aim of this section is to investigate the influence of the page
size and the order in which the two dimensions are indexed.  We use
two datasets with 10 million intervals each: the first has uniform start times and
skewed durations (Zipf distribution), the second has skewed start
times and uniformly distributed durations.  As for the query workload,
batches of 10\,000 \qrd, \qro, and \qdo queries are considered.  The
page size is varied between 1 and 10\,000.

Due to the size of the dataset, any page size $s$ above
$\sqrt{n} \approx 3\,162$ results in a degenerate configuration, where
there is a single column (of size $s^2=n$) that contains all the
intervals. In this case the intervals are partitioned according to a
single dimension, with the number of cells controlled by the page size
parameter.  This situation might also occur at lower values of the
page size parameter, depending on the number of distinct values in the
dimension being partitioned.  On the other hand, for page size $1$,
each cell of the grid contains only intervals with the same start time
and duration, and it contains all of them.  Querying the data
structure in this case amounts to perform binary searches directly on
the values of the domains of start times and durations.

Figure~\ref{fig:param-dependency} reports the results of this
experiment in terms of queries per second. To ease the comparison
between the plots, we rescale the throughput by the highest value for
each combination of dataset and query workload.  Degenerate
configurations resulting in a single column are reported as triangles
rather than dots.
We observe very different trends for different query workloads.

Consider first the dataset with skewed durations.  For \qrd queries,
both orderings of dimensions exhibit a similar behavior. The peak
performance is reached by intermediate values around the page size.
If we consider \qdo queries the profile changes. Indexing first by
duration (blue line) slightly favors smaller page sizes, which imply
smaller columns and thus a more fine grained access to the data.
Indexing first by time, instead, requires a \qdo query to traverse all
the columns.  In this scenario high page sizes are favored, since they
translate into fewer columns to be iterated through.  For \qro queries
we observe a symmetric pattern.

When data has skew on the start times, the patterns exhibited by the
two indexing orders are rather different.  First, we note that the gap
between the best configurations of the two indexing orders for \qrd queries is much
wider than with the other dataset.  Second, while
indexing first by duration exhibits a similar pattern on both
datasets, indexing first by start time performs best in the degenerate
cases of a single column, i.e., with no partitioning of the start
times at all.  This is a consequence of
Observation~\ref{obs:dimension-order}. In particular, having skewed
start times exacerbates the difference in selectivity between the
range and duration constraints: when start times are very concentrated
on the timeline, the range constraint of a query is satisfied by
either most of the intervals or by almost none.

Overall, we observe that indexing first by duration has either better
or comparable performance compared to indexing first by time.
Therefore, we recommend to choose the former ordering of dimensions
when building an \ours.

\begin{figure}
  \includegraphics[width=\columnwidth]{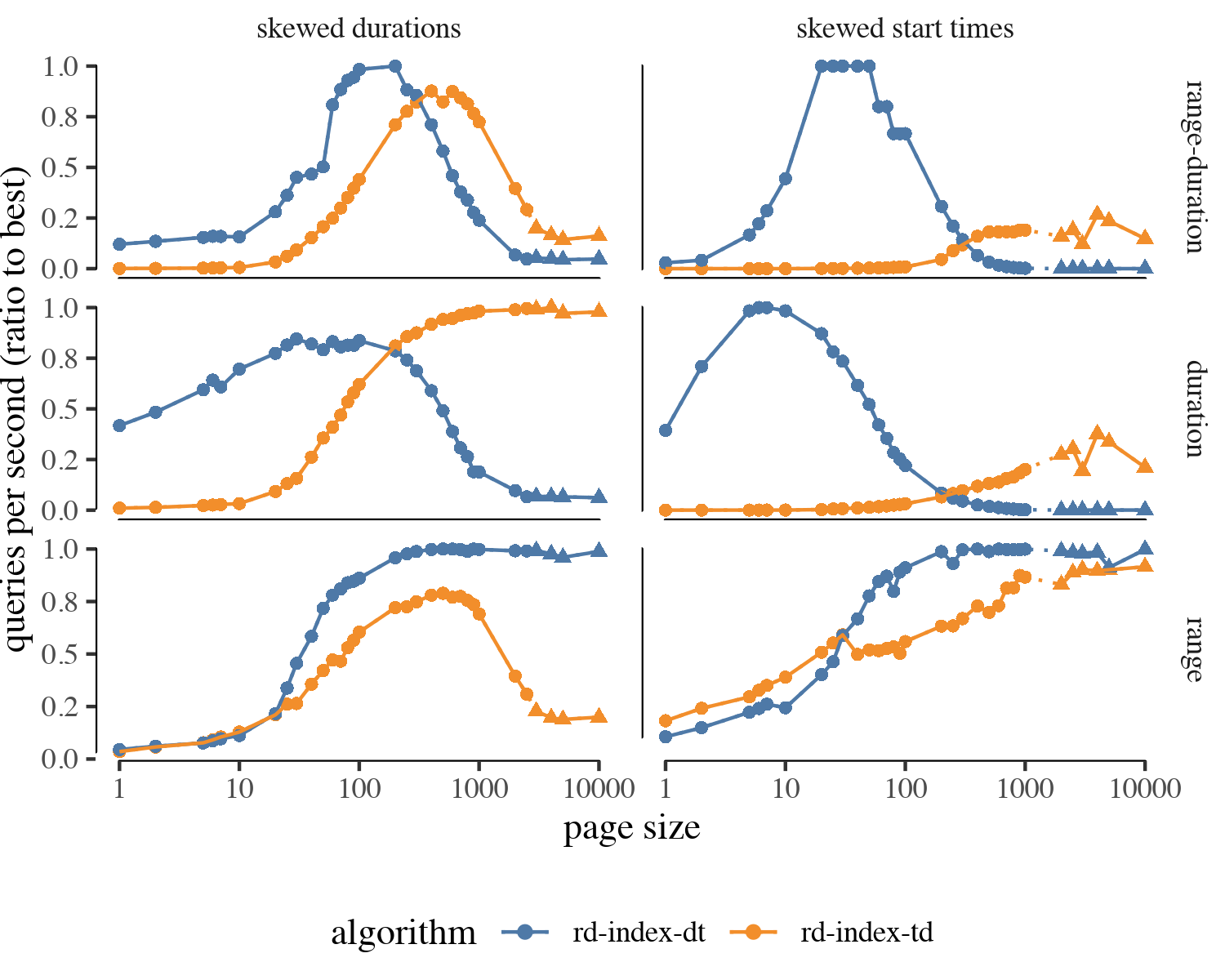}
  \vspace*{-2em}
  \caption{Dependency of the performance, in terms of queries per second, on the
    page size parameter.
    \label{fig:param-dependency}}
  \Description{}
\end{figure}

\subsection{Scalability with Respect to the Input Size}
\label{sec:scalability}

To test the scalability of the index structures, we consider \qrd
queries and datasets of increasing size, while maintaining the output
size constant.  This allows to assess the influence of the input size
on the performance without conflating the results with the time
required to iterate over larger outputs.

We consider the three real-world datasets, along with a synthetic one
with uniform start times and Zipf-distributed durations, which is
representative of the distribution of many real-world datasets.  Then,
we artificially increase their size as follows.  Given a dataset and a
scale parameter $\eta$, let $\bar{t}$ be the span of time covered by
all the intervals in the dataset.  We make $\eta$ copies of each
interval and shift copy $i$ in time by $i\cdot \bar{t}$, for
$i\in[0, \eta)$.  The underlying idea
is to repeat the temporal patterns of the dataset on a longer time
scale, simulating the scenario in which the relation grows over time.

Figure~\ref{fig:scalability} reports for each scale factor the
performance of the best configuration of each algorithm.  First, we
note that in general the relative performance of the data structures
does not change at different dataset scales.
There are some notable exceptions. The performance of \btree degrades
by a factor $\approx$10 from scale 1 to scale 10. The reason is that
the \btree indexes the durations, and under our synthetic construction
the number of intervals associated to each duration increases by the
same scale of the dataset.  For similar reasons, the performance of
\rddt degrades, albeit in a much less pronounced way.

On Webkit, the performance of \gfile and \pindex increases with the
scale as the dataset: the effect of our synthetic construction in this
case is to compensate for the skew in the start times, giving to both
data structures the chance of better partitioning the time dimension.

\begin{figure}
  \includegraphics[width=\columnwidth]{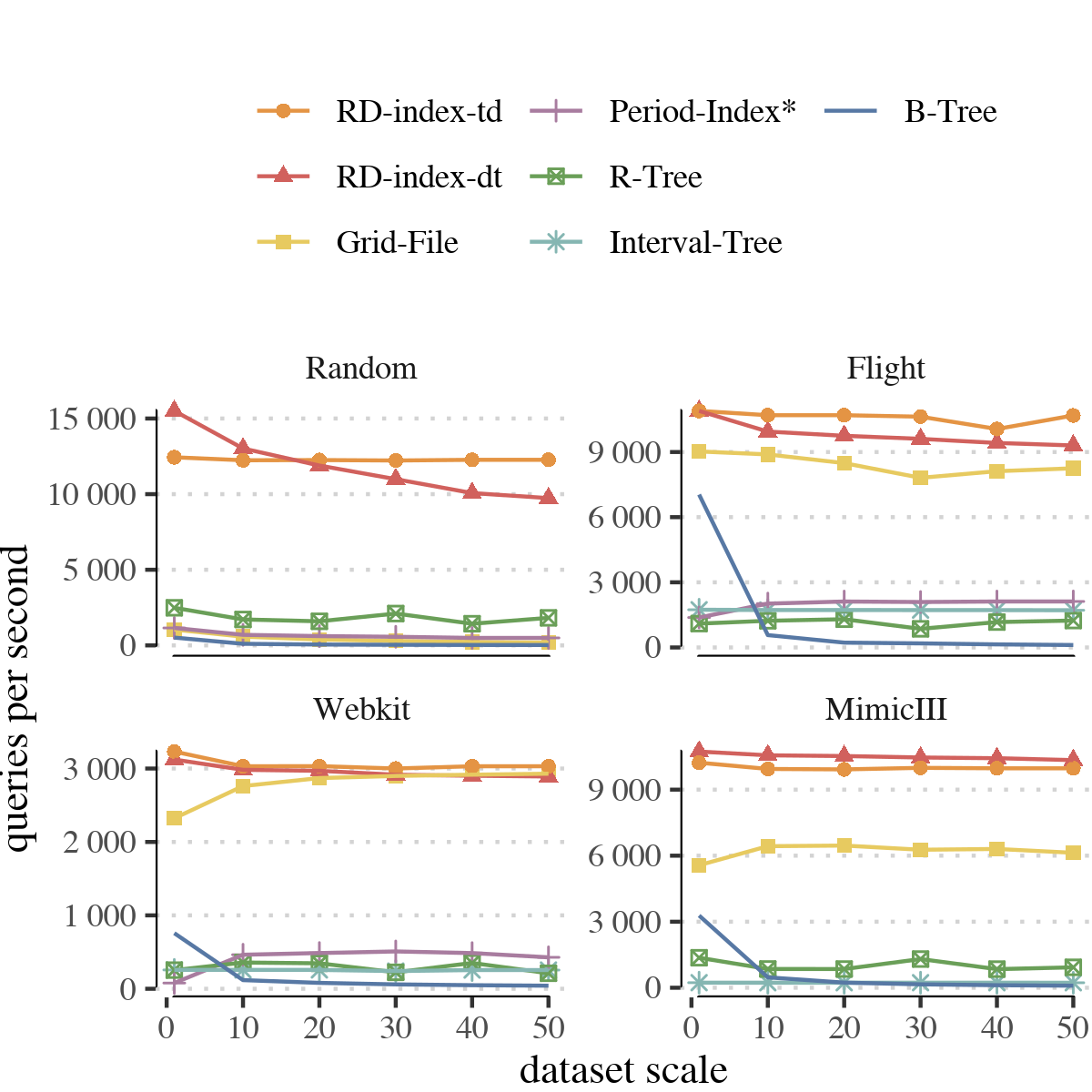}
  \caption{Scalability of the index structures for the increasing dataset sizes.
    Performance is measured in queries per second.
    Each point corresponds to the best configuration of the given data structure for a particular dataset scale, up to 50 times the original size.
    The workload is of range-duration queries. 
    \label{fig:scalability}} \Description{Scalability figure, showing
    that our proposed approach has superior performance for all tested
    dataset sizes, and scales gracefully.}
\end{figure}

\subsection{Insertion Performance}
\label{sec:insertion}

\begin{figure*}
  \begin{center}
    \includegraphics[width=\textwidth]{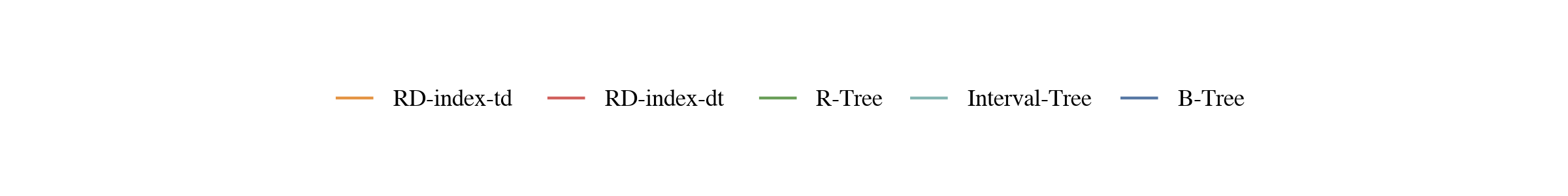}
  \end{center}
  \begin{subfigure}{0.49\textwidth}
    \includegraphics[width=\columnwidth]{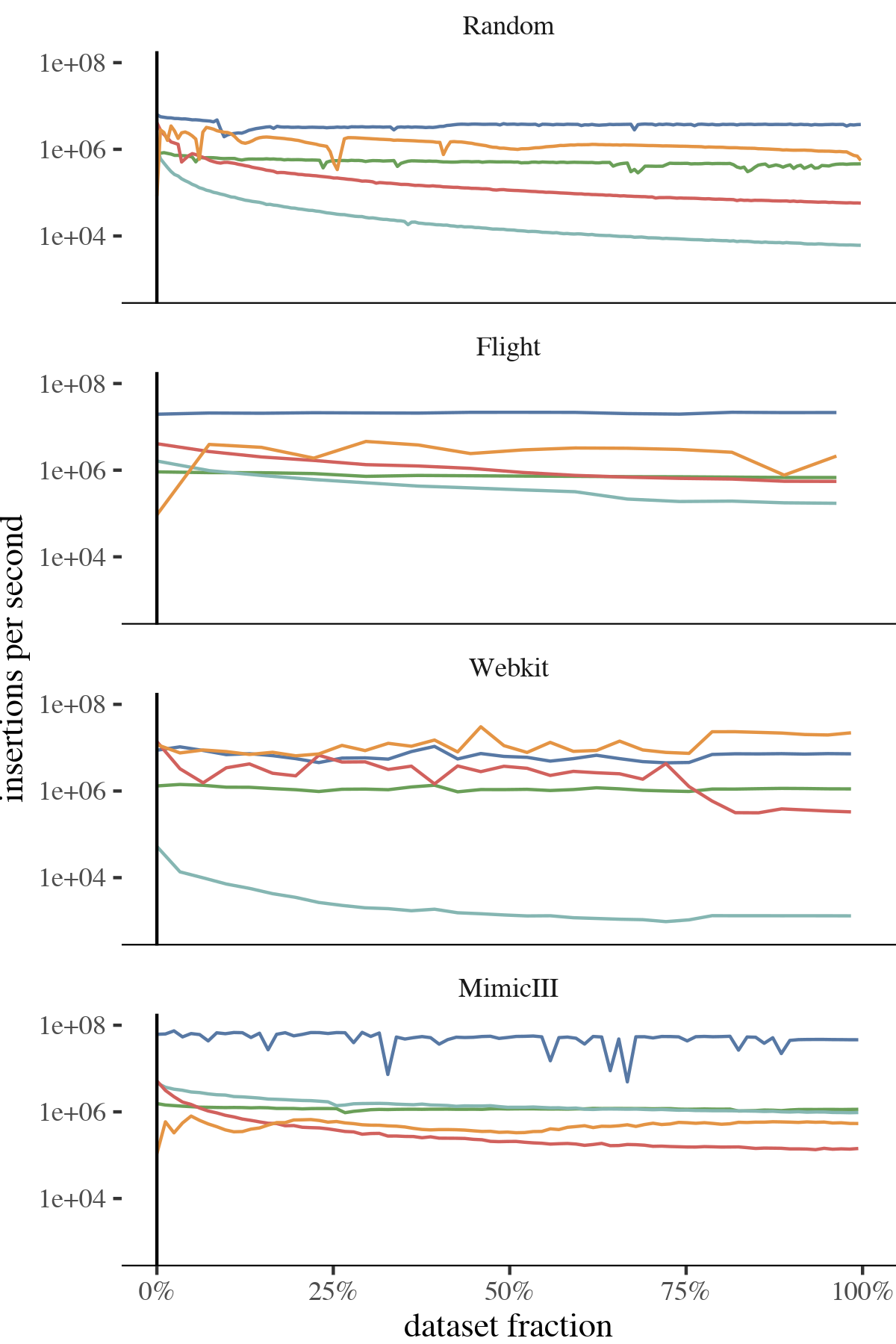}
    \subcaption{Intervals inserted in random order}
  \end{subfigure}
  \begin{subfigure}{0.49\textwidth}
    \includegraphics[width=\columnwidth]{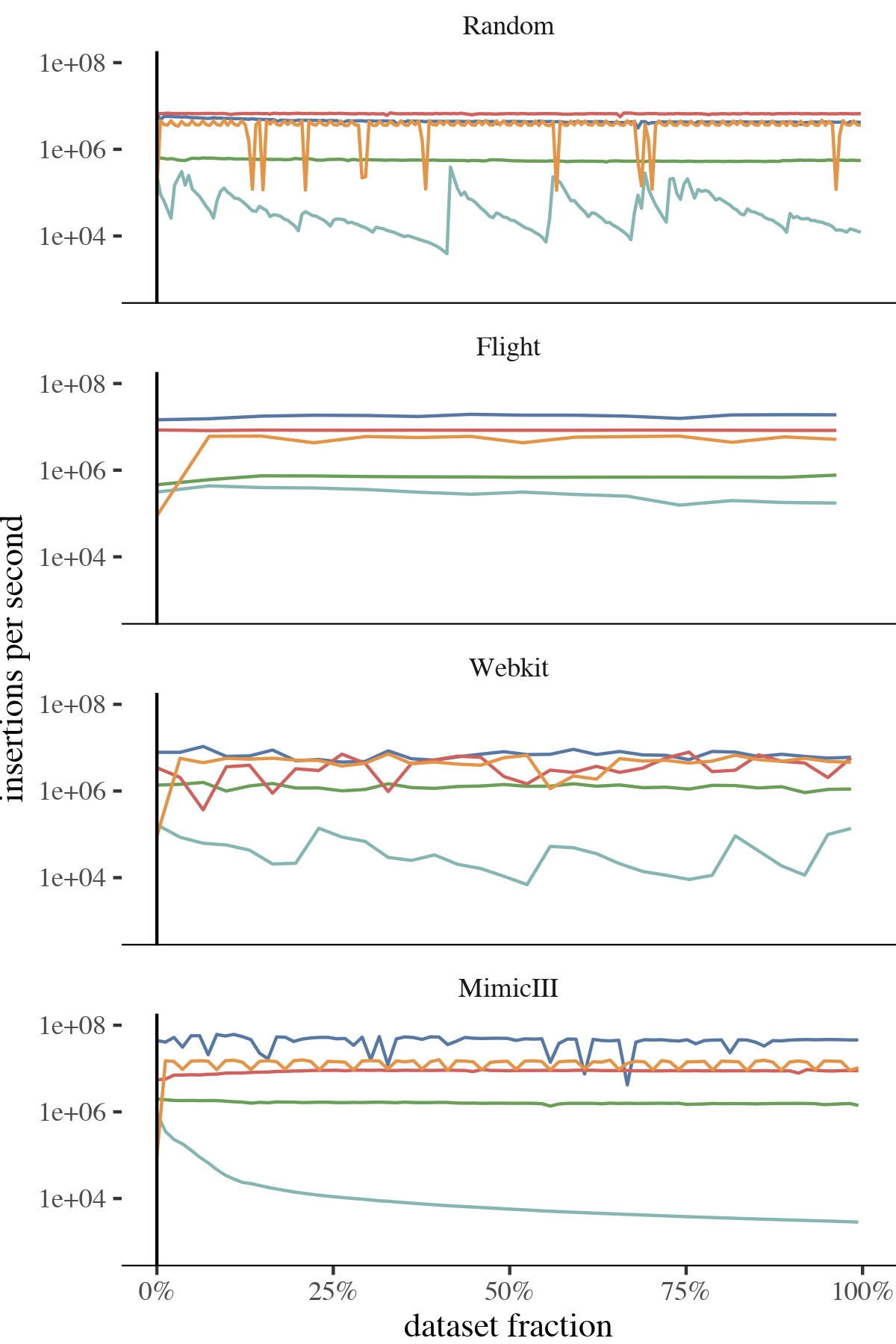}
    \subcaption{Intervals inserted by increasing start time}
  \end{subfigure}
  \caption{\label{fig:insertion}Performance of the insertion for
    different index structures.  The left group of plots (a) reports the
    throughput of insertions when intervals are added in random order.
    The right group of plots (b) focuses on insertions by increasing
    start time, i.e.\ a \emph{append only} scenario.}
\end{figure*}

We now focus on the insertion operation, considering \ours, \btree,
\rtree, and \itree.  We omit from the comparison \pindex, which does
not support updates, and \gfile, which is a static data structure that
requires to know the range of the data beforehand.

For each of the four datasets we considered in the previous sections,
we insert intervals into initially empty indices.  The expectation is
that the insertion performance degrades as the index grows larger. To
measure this effect, we insert the intervals in batches of 50\,000,
measuring the time for each batch in order to be able to estimate the
throughput of the insertions as the size of the index grows.  We
perform two sets of experiments.  In the first set of experiments, the
intervals are inserted in random order.  In the second set of
experiments, intervals are inserted by increasing start time, which
simulates a natural \emph{append only} scenario for time-related data.

Figure~\ref{fig:insertion} reports the results of these experiments.
The x-axis reports the fraction of the dataset that has been inserted
into the index.  The y-axis reports, in logarithmic scale, the number
of insertions per second.
In most cases, the best performing data structure for the insertion
workload is the \btree, both when data is presented in random and
sorted order.  \ours follows on the second place for most datasets,
with the ordering first by time and then by duration usually
performing better.  For randomly-ordered insertions, we note that the
performance of \ours tends to slightly decrease as the index size
increases.  In the more realistic append only scenario, instead, the
insertion throughput of \ours is more stable and tends to remain
constant over time.  This is expected, since in such a setting only
the last column (when the start time is the first dimension being
indexed, otherwise the last cell of each column) is ever restructured,
requiring very little data to be moved.  Furthermore, the performance
in this scenario improves compared to the random order of insertions,
in particular on the \dataset{MimicIII} dataset, and is on par with
the \btree on all datasets.

\section{Conclusions}
\label{sec:conclusions}

\ours is an index data structure for temporal intervals that allows to answer
efficiently range-duration queries.  Our approach, which has provable
theoretical guarantees, lends itself to a simple and efficient implementation.
In particular, its ability to adapt to the distribution of the input data makes 
it compare favorably with the state of the art on a variety of workloads.
In particular, \ours has superior performance on a vast array of mixed workloads.

A direction of future work is to extend the \ours to support interval
joins~\cite{DBLP:journals/vldb/PiaDHP21}, thus addressing several needs with a
single index.
Furthermore, the favorable comparison with the \rtree suggests that a
promising research direction is to extend the ideas on which \ours is
based to the case of multidimensional spatial data.

\bibliographystyle{abbrv}
\bibliography{references}

\appendix
\section{Additional Figures}
\label{sec:additional-figures}

In this appendix we report on the complete results for mixed workloads of which Figure~\ref{fig:mixed} in Section~\ref{sec:mixed} is a compact summary.

Figure~\ref{fig:tradeoff-all} reports a collection of ternary plots, each focusing on a particular combination of dataset (columns) and index structure (rows).
In order to accommodate on a readable color scale the wide range of throughputs encompassed by indices, we color-code the throughput divided by percentiles.
Therefore, brighter colors correspond to high throughputs, dark colors correspond to low throughputs.
Furthermore, these throughput percentiles are computed \emph{per dataset}.

Therefore, we can read Figure~\ref{fig:tradeoff-all} as follows.
For a given dataset (say \texttt{Random}) we can take two different algorithms (e.g. \ours and \btree) and compare the color in any given area of the two corresponding ternary plots.
For instance focusing on the bottom-right corner of the ternary plots, corresponding to workloads comprised mostly of \qrd queries, we observe that the plot for \ours features a brighter color than the one for \btree.
This means that the throughput of \ours is higher than the one of \btree for those workloads.

From Figure~\ref{fig:tradeoff-all} we can therefore notice that the behavior of \rddt and \rdtd is very similar, across all datasets.
Furthermore, we notice that with the exception of workloads where the vast majority of queries are \qdo, \ours enjoys a higher throughput than other indices.

\begin{figure*}
  \includegraphics[width=\textwidth]{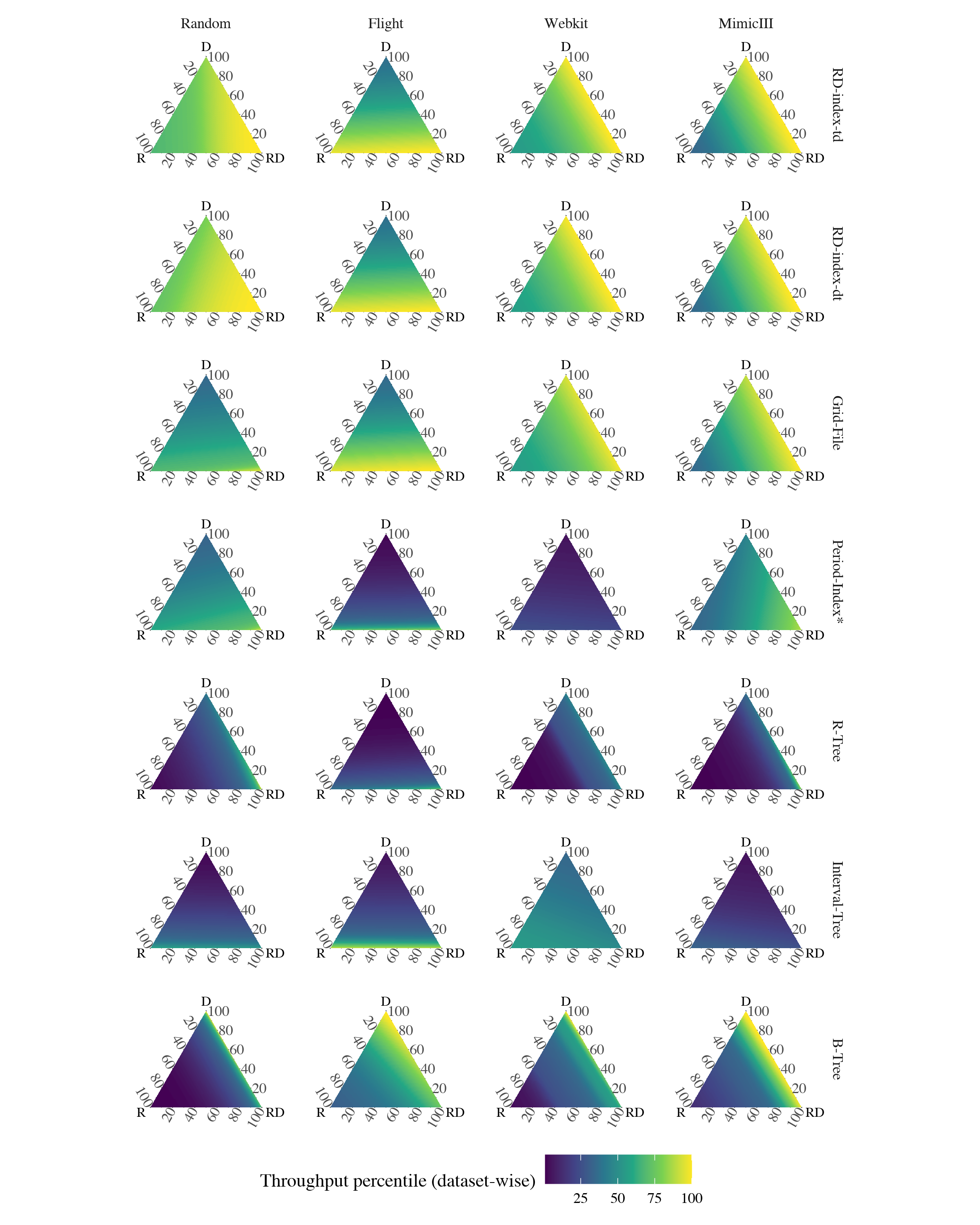}
  \caption{\label{fig:tradeoff-all}Performance of all index structures on all considered datasets, for different mixed workloads.}
\end{figure*}

\end{document}